%% file: arxiv-v2.tex
\documentclass[11pt]{article}

\usepackage[margin=2.5cm,]{geometry}
\usepackage{authblk}
\usepackage{glossaries}
\usepackage{amsmath}
\usepackage{amssymb}
\usepackage{amsthm}
\usepackage{dsfont}
\usepackage{physics}
\usepackage{siunitx}
\usepackage{graphicx}
\usepackage[hidelinks]{hyperref}
\usepackage{enumitem}
\usepackage{standalone}
\usepackage{framed}
\usepackage{tikz}
\usetikzlibrary{matrix,shapes,arrows,positioning,chains,calc,backgrounds,fit}

\numberwithin{equation}{section}
\newtheorem{theorem}{Theorem}[section]
\newtheorem{lemma}[theorem]{Lemma}
\theoremstyle{definition}
\newtheorem{definition}{Definition}[section]

\renewcommand{\Pr}{\operatorname{\mathrm{P}}}
\newcommand{\sender}{\textsf{Sender}}
\newcommand{\receiver}{\textsf{Receiver}}

\makeatletter
\def\blfootnote{\xdef\@thefnmark{}\@footnotetext}
\makeatother

\input{acronyms}

\title{Non-Interactive Oblivious Transfer and One-Time Programs\\ from Noisy Quantum Storage}

\author[ \hspace{-1ex}]{Ricardo Faleiro\thanks{\url{ricardofaleiro@tecnico.ulisboa.pt}}}
\author[ \hspace{-1ex}]{Manuel Goulão\thanks{\url{manuel.goulao@inesc-id.pt}}}
\author[ \hspace{-1ex}]{Leonardo Novo\thanks{\url{leonardo.novo@inl.int}}}
\author[ \hspace{-1ex}]{Emmanuel Zambrini Cruzeiro\thanks{\url{emmanuel.cruzeiro@lx.it.pt}}}
\affil[$*$]{\footnotesize Instituto de Telecomunicações, University of Aveiro, Portugal}
\affil[$\dag$]{INESC-ID, Instituto Superior Técnico, Universidade de Lisboa, Portugal;\newline Okinawa Institute of Science and Technology Graduate University, Okinawa, Japan}
\affil[$\ddag$]{International Iberian Nanotechnology Laboratory (INL), Portugal}
\affil[$\S$]{Instituto de Telecomunicações, Lisbon, Portugal;\newline Departamento de Engenharia Electrotécnica e de Computadores, Instituto Superior Técnico, Portugal}

\date{}

\begin{document}

\maketitle
\blfootnote{The authors are listed in alphabetical order.}

\begin{abstract}\noindent
Few primitives are as intertwined with the foundations of cryptography as Oblivious Transfer (OT). Not surprisingly, with the advent of quantum information processing, a major research path has emerged, aiming to minimize the requirements necessary to achieve OT by leveraging quantum resources, while also exploring the implications for secure computation. Indeed, OT has been the target of renewed focus regarding its newfound quantum possibilities (and impossibilities), both towards its computation and communication complexity. For instance, non-interactive OT, known to be impossible classically, has been strongly pursued. 
In its most extreme form, non-interactive chosen-input OT (one-shot OT) is equivalent to a One-Time Memory (OTM). OTMs have been proposed as tamper-proof hardware solutions for constructing One-Time Programs --- single-use programs that execute on an arbitrary input without revealing anything about their internal workings.
In this work, we leverage quantum resources in the Noisy-Quantum-Storage Model to achieve:

\begin{enumerate}[leftmargin=*]
    \item \textit{Unconditionally-secure two-message non-interactive OT} --- the smallest number of messages known to date for unconditionally-secure chosen-input OT.
    \item \textit{Computationally-secure one-shot OT/OTM}, with everlasting security, assuming only one-way functions and sequential functions --- without requiring trusted hardware, QROM, or pre-shared entanglement.
    \item \textit{One-Time Programs} without the need for hardware-based solutions or QROM, by compiling our OTM construction with the~\cite{C:GKR08,TCC:GIS+10} compiler.
\end{enumerate}
\end{abstract}
\newpage

\section{Introduction}
    The field of quantum cryptography had its genesis with the concept of ``Conjugate Coding''~\cite{ACM:Wiesner83}.
    The same primitive would later be published as \gls{OT}~\cite{Rabin81}, and would expand to become one of the most relevant primitives in cryptography.
    \gls{OT} has different but equivalent formulations~\cite{C:C87}, with the most prominent one being 1-out-of-2 \gls{OT}~\cite{CACM:EGL85}.
    1-out-of-2 \gls{OT} is a simple protocol between two parties, the \sender{} and the \receiver{}, where the \sender{} has two input messages (\(x_0,x_1\)) and the \receiver{} has an input choice-bit \(y\) and outputs the message \(x_y\).
    This happens while the \sender{} remains oblivious to \(y\) and the \receiver{} remains oblivious to \(x_{1-y}\).
    In this work, the \sender{} messages are considered to be bits.
    Notably, a series of works established the impossibility of constructing unconditionally-secure \gls{OT} and \gls{BC} without any assumption~\cite{PRL:LC97,PRL:Mayers97,PRA:Lo97}.
    This ignited a research line focused on finding the minimal requirements to implement these primitives.

    \medskip
    
    Given the impossibility to construct unconditionally-secure \gls{OT}, some restriction must be introduced to its execution environment.
    Often, limitations to the computing power (e.g., computational hardness assumptions), or a restricted physical model (e.g., bounded/noisy memory, shared randomness) are introduced in the system to enable the desired functionality.
    Moreover, relevant results show that \gls{OT} may be built from quantum computation and communication and (quantum-secure) \glspl{OWF}~\cite{EC:GLSV21,C:BCKM21}, or even weaker EFI pairs~\cite{ITCS:BCQ23}, thus relaxing the classical-world requirements of \gls{PKC}~\cite{STOC:IR89}.
    This opens up a series of new possibilities for potential \gls{OT} constructions, in particular, constructions that achieve otherwise unattainable security or efficiency levels.
    Indeed, low communication complexity is a highly desirable property in secure computation, and  following this research line, this work proposes to answer the question:
    \begin{center}\em
        What is the minimal number of communication rounds required\\ to construct 1-out-of-2 \acrlong{OT}?
    \end{center}
    
    Without further analysis, given that no restrictions are known on the minimal number of messages, achieving one-shot \gls{OT} would be the best one could aim for.
    However, classically, two-message \gls{OT} (one message each way) is optimal, as the messages of the \sender{} must somehow depend on the choice of the \receiver{}, or otherwise it could recover both messages, i.e., the protocol must be interactive.
    The pursuit of this two-message optimality has led to an extended research road (e.g., ~\cite{SODA:NP01,C:PVW08,EC:DGH+20}).
    Additionally, prior work has also explored delegating the \gls{OT} functionality to a trusted device, known as the \gls{OTM}~\cite{C:GKR08}, which can be transmitted between the parties, thus enabling non-interactive \gls{OT}, assuming the device remains secure.
    Therefore, it is pertinent to study what happens when quantum computation and communication and quantum-secure computational assumptions are introduced.
    
    We remark that, for the purposes of this work, a message means a single package of information sent from one party to the other.
    Thus, one-shot means that just one message is sent from one party to the other, as a single event.
    On the other hand, non-interactivity is used to state that communication is unidirectional, with one party sending possibly multiple messages to the other, which does not reply.%
        \footnote{Note that \textit{non-interactivity} is used with various different meanings in the literature, such as parties exchanging messages but the messages not depending on each other, e.g., non-interactive key-exchange.}

    \medskip

    Remarkably, such a simple primitive as \gls{OT}, by itself, is complete for general secure computation (\gls{2PC} and \gls{MPC})~\cite{FOCS:Yao86,STOC:GMW87,STOC:Kilian88}. 
    Consequently, a further line of investigation that analyzes how to relate the complexity of \gls{OT} with the complexity of \gls{MPC} has been pursued (e.g.,~\cite{EC:BL18,EC:GS18}).
    Since these only account for \gls{OT} built from classical resources, such works only aim for optimality as two messages of interactive communication (albeit sometimes with limited interaction, e.g.,~\cite{EC:IKO+11}).
    In~\cite{C:GKR08}, the concept of \gls{OTP} is introduced as program that can be run on an arbitrary single input, but only once, without revealing anything about itself (besides the output).
    The possibility to copy the program and running it again makes this primitive impossible to achieve, and to circumvent this issue,~\cite{C:GKR08} considers a secure memory device called \gls{OTM}.
    Interestingly, an \gls{OTM} is simply a rename of \gls{OT} with the added restriction of being a non-interactive, 1-message primitive (i.e., one-shot \gls{OT}), envisioned not as a cryptographic protocol (like \gls{OT}), but as a physical tamper-proof hardware device.
    (This distinction has since been blurred in the literature.)
    Then, as \gls{OT} yields \gls{2PC}, so does \gls{OTM} yield \gls{OTP}, with a compiler being able to achieve non-interactive malicious-secure \gls{OTP} from \gls{OTM}.
    The concept of \gls{OTP} was further studied and generalized in~\cite{TCC:GIS+10}, where it is shown that by allowing the parties to exchange tamper-proof hardware it is possible to directly build non-interactive secure \gls{2PC}.
    This leads to the question:
    \begin{center}\em
        How to design \acrlongpl{OTP} under weaker requirements and without trusted hardware?
    \end{center}

\subsection{Contributions}
    Three major conceptual contributions are established in this work, related to the proposed questions above.
    These are in the form of two 1-out-of-2 \gls{OT} protocols secure against malicious adversaries, and, as a corollary, a construction for \glspl{OTP}.
    As far as we are aware, these are the first evidences of such protocols in the scientific literature.
    \medskip
    
    The first conceptual contribution answers that
    \begin{center}\em
        Two-message non-interactive 1-out-of-2 \acrlong{OT} is possible\\ in the \acrlong{NQSM}, unconditionally.
    \end{center}
    This solution exploits the \gls{NQSM}, a model where the quantum memory of the parties performing the protocol, in particular the adversarial parties, is imperfect, and subject to noise.
    Thus, it prevents the indefinite (time) storage of quantum states, while no restrictions are made to the computing power or classical memory of the parties.
    Meanwhile, to execute the protocol, the honest parties require no quantum memory whatsoever.
    This is usually considered a general and weak assumption, as it is a realistic model that replicates the physical limitations of the present and near-future technology.
    Another construction is also provided by replacing the \gls{NQSM} by the stronger assumption of the \gls{BQSM} to substantially improve efficiency and remove artificially introduced time-delays.
    \medskip
    
    The second conceptual contribution evidences a
    \begin{center}\em
        One-shot 1-out-of-2 \acrlong{OT} (or, \acrlong{OTM}) in the \acrlong{NQSM} assuming the existence of  a \acrlong{OWF} and a \acrlong{SF}.
    \end{center}
    This construction again relies on the \gls{NQSM}, but also depends on the existence of a quantum-secure \gls{OWF} and the existence of a \gls{SF}, as the construction relies on the primitive of \gls{TLP}.
    The existence of \glspl{SF} (also called non-parallelizing languages~\cite{ITCS:BGJ+27}), and their relation to the construction of \glspl{TLP} have been previously studied~\cite{ITCS:BGJ+27,IC:JMRR21}, while candidates for \glspl{SF} ranging from hash functions (\gls{QROM})~\cite{EC:CFHL21} to lattice-based assumptions~\cite{C:LM23,C:AMZ24} have been recently proposed in the literature.
    Again, no restrictions are made to the classical memory of the parties, and no quantum memory is required to honestly complete the protocol.
    But now, it must be assumed that the parties are probabilistic-polynomial-time quantum machines and have limited computing power.
    Another technical contribution from this construction is the introduction of the use of the \gls{TLP} primitive when proving security in the \gls{NQSM}, such that the time it takes to solve the \gls{TLP} enforces quantum decoherence of the memories of an adversary.
    
    \medskip

The third conceptual contribution exhibits a
    \begin{center}\em
        \acrlong{OTP}, secure against malicious adversaries, assuming the existence of \\ a one-shot \acrlong{OT}/\acrlong{OTM} in the \acrlong{NQSM}.
    \end{center}
Given that \gls{OTM} is impossible in the plain model, even under computational assumptions~\cite{C:BGS13}, we construct \glspl{OTP} in the \gls{NQSM}, achieving a solution in a weak trusted model --- arguably the best one can hope for.
We leverage our previous contribution, the one-shot \gls{OT}/\gls{OTM}, combined with the compiler from~\cite{C:GKR08} which transforms a malicious-secure \gls{OTM} into a malicious-secure \gls{OTP} for arbitrary functions.
For this corollary, we need the same assumptions as before of a quantum-secure \gls{OWF} and \gls{SF}, and work in the \gls{NQSM}.
\glspl{OTP} are an extremely powerful concept that is impossible to realize in most settings, with previous proposals relying on trusted hardware assumptions, which have often proven unreliable in practice. Here, we construct \glspl{OTP} under the weakest assumptions to date. \glspl{OTP} have applications such as software protection and one-time proofs~\cite{C:GKR08}.

\medskip 

Finally, in terms of technical novelty, a key contribution is the explicit handling of post-measurement classical information in non-interactive \gls{OT}, \gls{OTM}, and consequently \gls{OTP}. This approach enables security guarantees against the most general class of attacks conceivable by a malicious quantum adversary.
Although this framework~\cite{PRA:GW10} has been proposed for some time, to our knowledge, this is the first instance of its application to \gls{OT}, \gls{OTM}, and \gls{OTP}. The techniques employed are non-trivial and may have broader applications in quantum cryptography.
Our approach relies on deriving tight upper bounds for the eigenvalues of qubit register states, which, interestingly, are connected to specific Hamiltonians of spin-chain systems, widely studied in condensed-matter physics for entirely different reasons.
Additionally, a notable feature of our protocol is the potential use of entanglement to employ self-testing techniques, introducing an extra layer of (semi-)device-independent security.

\subsection{Related Work}\label{sec:relatedwork}

\paragraph{OT in Restricted Settings:}
    The impossibility of unconditional \gls{OT} from exclusively informational theoretical considerations demands extra assumptions, either physical or computational, beyond the validity of quantum mechanics~\cite{PRA:Lo97}.
    The first proposals for quantum-based \gls{OT} protocols were constructed in the same setup of the original Conjugate Coding~\cite{ACM:Wiesner83}, and were secure only in certain restricted settings~\cite{CK88,C:BBCS92}.
    In fact, the authors of~\cite{C:BBCS92} even described possible measurement attacks compromising \sender{}-security, wherein the \receiver{} would delay the measurements and implement multi-qubit measurements later on.
    Thus, in order to establish security, one should still require that the \receiver{} implements the measurements at the desired time, by any means necessary, e.g., computational or physical limitations.
    One alternative they propose is assuming the existence of commitment schemes secure against limited computing power, say using \glspl{OWF}~\cite{C:BBCS92}.
    In fact, a recent line of work has confirmed the belief that the weaker assumption of \glspl{OWF} suffices for secure \gls{OT} in the quantum world~\cite{EC:GLSV21,C:BCKM21}, as opposed to  the classical setting  where \gls{PKC} is known to be a requirement~\cite{STOC:IR89}.
    As an alternate approach, one may consider physically motivated restrictions, like bounding the memory of the adversaries, the \gls{BQSM}.
    This type of restriction had already been invoked in the classical setting~\cite{JC:Maurer92,C:CM97}, with explicit \gls{OT} constructions presented therein~\cite{FOCS:CCM98}, before being considered in the quantum setting \cite{FOCS:DFSS05} (only bounding the total quantum storage), and further generalized to a more realistic scenario \cite{PRL:WST08} (unbounded quantum storage, but noisy).
    Precisely, in \cite{PRA:S10} the \gls{NQSM} was explicitly leveraged  in order to prove the security of \cite{C:BBCS92}.
    Recently, constructions leveraging physical restricted models (\gls{BQSM}) together with computational assumptions (Learning-With-Errors) have been proposed, opening up a wide range of new applications and enabling device-independent \gls{OT}~\cite{NJP:BY23}.

\paragraph{Non-Interactive OT and MPC:}
    The \gls{OT} protocols proposed in this work are non-interactive, in the sense that communication is always one-way, from \sender{} to \receiver{}.
    For these kinds of \gls{OT} protocols, perfect \receiver{}-security can be immediately established from reasonable physical principles, like the \textit{no-signalling-from-the-future}~\cite{PRA:CDP09}. In fact, Wiesner's original proposal of Conjugate Coding~\cite{ACM:Wiesner83}, even if not proven to be secure for the \sender{}, was non-interactive, and thus perfectly \receiver{}-secure. 
    It follows naturally that physically constrained models precluding unbounded quantum storage, such as the \gls{BQSM} and \gls{NQSM}, would be prime candidates for constructing such non-interactive unconditionally-secure \gls{OT} protocols. Indeed, in~\cite{FOCS:DFSS05}, where the \gls{BQSM} was first introduced, a construction for non-interactive All-or-Nothing \gls{OT} based on the original Conjugate Coding setup was introduced. This was followed by a non-interactive two-message 1-out-of-2 Random \gls{OT}~\cite{C:DFR+07}, also in the \gls{BQSM}, which was further generalized to the \gls{NQSM}~\cite{PRL:WST08}.%
        \footnote{Translating Random \gls{OT} to chosen-input \gls{OT} requires one extra message~\cite{C:DFSS06}.}
    Lastly, different attempts to achieve secure computation non-interactively have been developed, e.g.:
    the subject of ``Non-Interactive Secure Computation''~\cite{EC:IKO+11,TCC:BL20}, a \gls{2PC} scenario that can be computed in two steps, but one step is delegated to a pre-processing publishing phase (which in practice makes it interactive); or similarly, the ``Private Simultaneous Messages'' protocols~\cite{STOC:FKN94,C:BGI+14,AC:HIJ+17,TCC:HIKR18}, where some parties communicate to a different entity a message that depends on their input, but requires that they share some randomness source (again delegating interaction to a pre-processing phase).

\paragraph{Round-Optimal OT and MPC:}
    Generally, the usefulness of round-optimal \gls{OT} is drawn from trying to improve the communication round complexity of \gls{MPC}.
    As its most costly primitive, and often reliant on \gls{PKC} primitives, minimizing the round complexity of \gls{OT} is paramount.
    In spite of this fact, while round-optimal \gls{OT} has been a pursued goal for a long time, research on the topic has mostly been restricted to classical solutions for \gls{OT}.
    Therefore, round-optimality of \gls{OT} is largely and explicitly been considered to be two messages, and necessarily interactive~\cite{SODA:NP01,EC:K05,C:PVW08,EC:DGH+20,AC:CSW20}.
    Furthermore, the round complexity of secure computation protocols has also been extensively studied in the literature.
    In particular, analyzing black-box constructions of \gls{MPC} in the plain model is known to require at least four messages, and interactivity ~\cite{C:KO04,EC:GMPP16,C:ACJ17,TCC:BHP17,C:HHPV18,TCC:RCG+20}.
    However, relaxations to the number of corrupted parties~\cite{C:IKP10,C:ACGJ18}, or to the security~\cite{C:KO04,FOCS:QWW18,TCC:ABJ+19,EC:COWZ22} allow for more efficient protocols to be achieved that only take two messages.
    Also, assuming shared randomness, it is possible to compile \(n\)-message \gls{OT} into \(n\)-message \gls{MPC}, for \(n\geq 2\)~\cite{EC:BL18,EC:GS18}.
    
    \paragraph{Hardware and software OTM and OTP:}
   
    Considering \gls{OT} and secure computation, together with non-interactivity, the notions of \gls{OTM} and \gls{OTP} naturally arouse~\cite{C:GKR08}.
    The pursuit of these constructions has continued since, originally focusing on secure hardware but later seeking to eliminate this strong requirement.
    In~\cite{ISC:DDKZ14}, it was shown how to implement \gls{OTP} without assuming the security of hardware devices, by relying on two trusted models that assume bounded leakage and bounded storage.
    In~\cite{TCC:EGG+22}, another hardware functionality (albeit more practical and available, named the counter lockbox) is used to devise \glspl{OTM} and \glspl{OTP}.
    A clear direction to approach the issue of copying the program and running it again is to explore the no-cloning property of quantum information.
    However, in~\cite{C:BGS13}, it is shown that \glspl{OTP} for all programs cannot exist in the plain model (even under computational assumptions), requiring some trusted setup.
    Still,~\cite{C:BGS13} generalized the concept of \glspl{OTP} to quantum programs.
    In~\cite{Q:BGZ21},  secure hardware (but stateless) is once again used to achieve quantum \glspl{OTP} assuming a bounded number of adversarial queries.
    
    \medskip

    Recently, in~\cite{ITCS:Liu23} an independent proposal for a \gls{OTM} was introduced.
    The crucial distinction with our approach lies in the way one limits said time-frame. Our \gls{OTM} leverages the \gls{NQSM} to limit the adversary's ability to maintain coherent quantum states over time due to noise.
    In \cite{ITCS:Liu23}, instead, an artificial bound on circuit depthness for a pre-fixed polynomial depth is considered (motivated by NISQ computers).
    On the other hand, our setup in the \gls{NQSM} is based on the physical difficulty of obtaining high-fidelity coherent quantum memories, a weak and appealing security model often considered nowadays for \gls{OT} and secure computation, e.g.,~\cite{PRR:LPAK23}.
    Moreover, in~\cite{ITCS:Liu23} the security proof requires the \gls{QROM}, hinging on idealized oracle behavior and leading to the usual pitfalls of heuristic oracle modeling.
    In terms of efficiency, both protocols scale polynomially in the computational security parameter, i.e., as the adversarial power grows (polynomially), the communication also grows (polynomially). 
    Another critical difference, regarding the adversarial model,~\cite{ITCS:Liu23} imposes the strong assumption that adversaries only may do full quantum measurements over the entire qubit register (meaning they do not need to consider post-measurement information).
    We consider arbitrary measurements and with classical post-measurement information, which constitutes the broadest class of attacks possible, coherent attacks, safeguarding even against cleverly chosen arbitrary partial measurements.
     
    \medskip
    
    Also recently, in~\cite{EC:ABKK23,C:BKS23}, the authors have tackled similar questions relating to non-interactive \gls{OT} (without tackling \gls{OTM} and \gls{OTP}).
    In~\cite{EC:ABKK23}, three constructions for \gls{OT} are presented.
    While the first solution claims to be one-shot, it assumes shared maximally-entangled pairs before the execution of the protocol, in a setup phase that is not accounted for as a round.
    This means that this protocol needs, effectively, two messages.
    This fact is explicitly acknowledged by the authors in their Section~2.2~\cite{EC:ABKK23}, where it is mentioned that to construct a protocol without assuming the setup phase, one more message must be introduced in an interactive manner (i.e., in the other direction).
    Also, their construction is for Random \gls{OT} in the \gls{QROM}, in opposition to this work, where a chosen-input \gls{OT} in the \gls{NQSM} is proposed.
    In~\cite{C:BKS23}, another construction for a one-shot \gls{OT} in the shared EPR pairs model is proposed, under the sub-exponential LWE assumption, together with a proposal for an \gls{MPC} protocol also requiring the QROM.
    In contrast, the contributions of our work rely on the \gls{NQSM} to show that not only an unconditionally-secure two-message non-interactive chosen-input \gls{OT} exists, but even a one-shot \gls{OT} from \glspl{OWF} and \glspl{SF} and without the need to have pre-shared entanglement; furthermore, \glspl{OTP} is built from these constructions.

\subsection{Overview}
    Here, an overview of the main results is provided.
    The objective is to give intuition about the contributions of this work in a simple manner. 
    As such, most of the arguments reasoning is built-up from well-known principles of quantum information and adapting them to the desired setting.

\subsubsection{Non-Interactive Oblivious Transfer and One-Time Memory}
    Constructing non-interactive \gls{OT} and \gls{OTM} has been an elusive task, known to be impossible classically, without any further computational assumptions and trusted models.
    Moreover, not only unconditionally-secure \gls{OT}, but even \gls{OT} from \glspl{OWF} were widely held to be impossible.
    From recent results, it is now known that \gls{OT} and \gls{MPC} are possible from quantum computation and information and \glspl{OWF}, without the need for \gls{PKC}.
    Also, unconditionally-secure \gls{OT} can be enabled by restricting the physical setting of its execution, specially interesting for realistic physical models.
    For the particular case of \gls{OTM} (i.e., one-shot \gls{OT}), its realization is impossible in the plain model even under computational assumptions~\cite{C:BGS13}.
    
    Two relevant \gls{OT} constructions are provided, based on the realistic modeling of imperfect quantum memories, the \gls{NQSM}:
    \begin{itemize}
        \item Non-interactive two-message unconditionally-secure (chosen-input) 1-out-of-2 \gls{OT}, secure against malicious adversaries.
        \item \gls{OTM}, i.e., one-shot (chosen-input) 1-out-of-2 \gls{OT}, assuming the existence of a \gls{OWF} and a \gls{SF}, secure against malicious adversaries.
    \end{itemize}
    
    The first proposed \glspl{OT} attains unconditional security, and is conceptualized in the \gls{NQSM} so as to avoid the usual impossibility results.
    The \gls{NQSM} is a highly appealing model, as physical quantum memories are imperfect and suffer from quantum decoherence relatively fast, and is specially relevant as the protocol does not require any memory to honestly run.
    The protocol works as follows:
    \begin{enumerate}
        \item The \sender{} prepares two maximally-entangled qubits in which it encodes its inputs.
        \item The \sender{} hides these two qubits in a large set of uniformly random qubits, such that the \receiver{} cannot tell which qubits encode the information.
        \item The \receiver{} measures each qubit, in a basis defined by its input-bit, and stores the measurement results.
        \item After waiting some pre-defined time, the \sender{} communicates the encoding, which allows the \receiver{} to compute the desired \gls{OT} output.
    \end{enumerate}
    To ensure security, the \gls{NQSM} establishes that, after some time, the quantum memory of the parties becomes irretrievable.
    So, this model is leveraged by making an adversary trying to break the protocol wait a predefined amount of time, such that it cannot make joint measurements on only the information qubits unless it guesses them correctly (as from separate measurements cheating is impossible).
    This can be made to be unfeasible by appropriately choosing the amount of hiding qubits that the \sender{} sends to the \receiver{}.
    Here, the distance between the statistical distributions become closer at a linear rate, and can be made sufficiently close by appropriately setting the statistical security parameter.
    Also, state preparation is very efficient and current technology already enables high-rate sending of qubits.
    Moreover, from the non-interactivity of the protocol, the \sender{} cannot do anything, as it is unable to extract information from future events.
    Evidently, in this construction, no quantum memory whatsoever is required for the honest parties to engage in the protocol.
    
    A variation of the first protocol is also proposed, where time efficiency is increased in exchange for replacing the weaker \gls{NQSM} with the stronger assumption of the \gls{BQSM}.
    Here, the \gls{BQSM} is exploited, as it allows for an instant to be chosen when the adversary can only store a subset of its total quantum memory.
    If this instant is chosen to be exactly between the \sender{} sending the qubits and sending the encoding, then, no waiting time is required to achieve security, given that a large enough number of qubits are sent to mask the legitimate ones.

    \medskip
    
    The second proposed \gls{OT} achieves the captivating goal of being one-shot, i.e., constitutes a \gls{OTM}.
    Here, for the first time, the \gls{NQSM} is connected with the concept of \gls{TLP}.
    Conveniently, a \gls{TLP} is a primitive that allows for a party to send a hidden message to another, such that the recipient must spend some time (via computation) to recover the concealed information.
    So, from the \gls{NQSM}, by requiring that an adversary must spend some physical time to gain information that would enable an attack, its memory storage suffers from the phenomenon of quantum decoherence, and the attack becomes unfeasible.
    In this particular construction, the same rationale from the previous one is used, where the information qubits are hidden among random qubits, such that an adversary cannot perform joint measurements.
    But here, the encoding is hidden inside the \gls{TLP} and sent together with the full state, and the parameters of the \gls{TLP}, i.e., the time it takes to solve it is chosen such that quantum decoherence would happen in the meantime.
    Thus, a malicious \receiver{} cannot store the qubits until it knows the encoding of which two to measure jointly and break security.
    This is essentially the same situation as in the previous two-message construction, but delegates the time-keeping from the \(\sender\) to a computational cryptographic primitive to achieve the \gls{OTM}.
    Moreover, using a \gls{TLP} does not give any advantage to a malicious \sender{} to learn the input of the \receiver{}.
    Hence, this \gls{OT} protocol immediately achieves everlasting security, meaning that the non-chosen message becomes perfectly irretrievable after the execution of the protocol, since the computational assumption only needs to hold while the \gls{TLP} is being solved.
    Clearly, also in this construction, no quantum memory whatsoever is required for the honest parties to engage in the protocol.

\paragraph{Remark.}
Our \gls{OT} proposal is secure both regarding indistinguishability of statistical distributions (via the hiding of the information in the qubits --- quantum communication), and the computational unfeasibility of a polynomially-bounded quantum adversary to prematurely extract a secret encoding (via the \gls{TLP} building block --- classical communication).
First, the computational security of the protocol must account for a growing computational power of any adversary (represented by a computational security parameter), meaning that the communicated \gls{TLP} size grows polynomially to the computational security parameter.
In opposition, the number of qubits that must be sent to ensure security does not grow with the power of the adversary (it is constant), and must be set to the desirable level of indistinguishability of the distributions according to the statistical security parameter.
We instantiate this statistical distance to \(2^{-40}\), meaning that the constant number of qubits sent must be \(2^{40}\) independently of the power of the adversary.
Therefore, in our first proposal (non-interactive \gls{OT}), the classical and quantum communication is constant in relation to the adversarial power (statistical security);
and, in our second proposal (\gls{OTM}), the quantum communication is constant and the classical communication grows polynomially with the computational security parameter, as does the size of our proposal for the \gls{OTP} construction.

\subsubsection{One-Time Program}
Influenced by the seminal work of Yao's garbled circuits~\cite{FOCS:Yao86}, \cite{C:GKR08} introduced the concept of \gls{OTP}.
An \gls{OTP} implements the functionality of a black-box that can only be evaluated once on an arbitrary input \(x\), and returns the value of a function \(f\) on this input.
Security of the \gls{OTM} requires that no adversary can learn more about \(f\) than what can be learned from a single tuple \((x,f(x))\).

In~\cite{C:GKR08}, a compiler is introduced that constructs a malicious-secure \gls{OTP} for an arbitrary function from a parallel-\gls{OTM}.
We show that our \gls{OTM} may be run in parallel without hindering security, particularly not allowing an adversary to query the \glspl{OTM} adaptively, fulfilling the requirements of a parallel-\gls{OTM}.
Thus, as a direct corollary of the compiler of~\cite{C:GKR08}, we achieve an \gls{OTP} from the constructed \gls{OTM} in the \gls{NQSM}, assuming the existence of a \gls{OWF} and a \gls{SF}.
We stress that our \gls{OTP} construction has the limitation that it must be run directly after being received due to working in the \gls{NQSM}, as qubits holding the information will suffer decoherence, making it impossible to run the \gls{OTP} after a long enough period of time.

\subsection{Open Questions}
    The statistical security component of our proposed constructions require a large number of qubits (exponential in the statistical security parameter) to ensure security.
    This happens as the security of the protocols rely on hiding information in a combinatorial manner with the number of sent qubits.
    While constant for a fixed statistical distance of the distributions and independent of any adversary power (even unbounded), this large number of qubits required hinders efficiency of the protocols, given the current available technology. 
    Hence, removing this requirement would be of relevance, even if replacing it by computational assumptions.
    Following the same approach that was used for this construction, one could try to leverage a relationship between higher-dimensional qudits and string \gls{OT}. 
    
    Moreover, device-independent security extends the standard notion of security, such that even the devices or laboratories used by the parties do not need to be trusted. Although this is a highly appealing security model, demonstrated for OT~\cite{NJP:KW16,NJP:BY23} and other cryptographic primitives~\cite{Nat:P10,PRL:VV14,NJP:AMPS16,PRA:FG21}, it is also extremely demanding, as it relies on the violation of Bell inequalities. To address this challenge, semi-device-independence relaxes the model by allowing certain assumptions to be made, while still preserving the essential properties of the quantum systems that ensure security. The use of entanglement in our construction makes it an attractive candidate for analysis within the device-independence framework. For example, introducing self-testing as a subroutine in some rounds of the protocol could partially verify the resources used, adding a layer of (semi-)device-independence to the security.

\section{Background}\label{sec:background}
\subsection{Quantum Systems, States, and Processes}
    A finite \(d\)-dimensional quantum system is represented by a Hilbert space \(\mathcal{H}\cong \mathbb{C}^d\). Of fundamental importance in quantum information is the \(2\)-dimensional quantum system \(\mathcal{H}\cong \mathbb{C}^2\), the \textit{qubit}.
    Composition of quantum systems is given by the tensor product of individual Hilbert spaces, such that a system of \(n\)-qubits, often called a \(n\)-qubit \textit{register}, is represented by \(\mathcal{H} \cong  \mathbb{C}_{(1)}^{2} \otimes \dots \otimes \mathbb{C}_{(n)}^{2}\cong\mathbb{C}^{2^n}\).
    
    The state-space of a quantum system is given by the set of all trace one, Hermitian, positive semi-definite operators acting on the corresponding Hilbert space, i.e., \(\rho \in \mathcal{L}(\mathcal{H}) \cong \mathbb{C}^{d \otimes d}\). Pure states can be described by outer products of vectors of the Hilbert space \(\rho = \ketbra{\psi}\) and, in that case, it is  customary to represent the state of the system by  the vector itself, \(\ket{\psi} \in \mathcal{H}\cong \mathbb{C}^d\). Pure states in composite systems are said to be \textit{entangled}, if they cannot be factorized into vectors of the product Hilbert spaces.  Also important are the  four different two-qubit (\(\mathbb{C}_{S}^2 \otimes \mathbb{C}_{R}^2 \))  maximally entangled states, known as Bell states,
    \begin{equation}
    \label{Bell_state}
        \ket{B_{xy}}_{SR}= \frac{1}{\sqrt{2}}\left(\ket{0y}+(-1)^x\ket{1\Bar{y}}\right)_{SR},
    \end{equation} 
    for  \(x,y\in \{0,1\}\), and \(\Bar{y}\) being the negation of \(y\). A two-qubit pair in any of the Bell states is said to form an \gls{EPR} pair (or Bell pair).
    
    In quantum-information processing, it is useful to adopt an operational perspective when describing the evolution of quantum systems throughout protocols. From that perspective, one considers different types of idealized black-box processes that can be implemented on quantum systems, changing their states at different stages.
    Fundamentally, three processes are noteworthy:
    
    \begin{itemize}
        \item \textit{Preparation} (classical-to-quantum process): Process with non-trivial classical input \(x\), which outputs a corresponding quantum state \(\rho_x\) obeying the usual normalization \(\Tr(\rho_x)=1\).
        
        \item \textit{Transformation} (quantum-to-quantum process): Process taking as input a state \(\rho_{\textup{in}}\) and outputting \( \rho_{\textup{out}}=\Phi(\rho_\textup{in}) = \sum_k E_k \rho E_k^{\dagger}\), for \(\Phi \in \left\{\mathcal{L}(\mathcal{H}_{\textup{in}}) \to \mathcal{L}(\mathcal{H}_{\textup{out}})\right\}\) a \gls{CPTP} map, and  \(\{E_k\}\) the corresponding Kraus operators satisfying \(\sum_k E_k^{\dagger} E_k = \mathds{1}\). For a unitary transformation \(U\) (\(U^{\dagger}U=UU^{\dagger}=\mathds{1}\)), it simplifies to 
        \(\rho_{\textup{out}}= U\rho_\textup{in} U^{\dagger}\). Transformations can also be considered to have a classical control-input whose value dictates the fixed transformation applied.

        Especially important in this work are the X, Y, Z Pauli unitaries and the Hadamard transform, given in matrix form, respectively, as 

\begin{align}
\label{Pauli}
    X &= \begin{pmatrix} 0 & 1 \\ 1 & 0 \end{pmatrix}, \quad
   Y = \begin{pmatrix} 0 & -i \\ i & 0 \end{pmatrix},\quad
    Z = \begin{pmatrix} 1 & 0 \\ 0 & -1 \end{pmatrix}, \quad   H = \frac{1}{\sqrt{2}}\begin{pmatrix} 1 & 1 \\ -1 & 1 \end{pmatrix}
\end{align}

        \item \textit{Measurement} (quantum-to-classical process): Process with non-trivial input tuple \((y,\rho)\) (classical \(y\) and quantum \(\rho\)), and a classical output \(m\). It is modelled by a  \gls{POVM} \( \{M_{m|y}\}_{m} \), such that, for input \((y,\rho)\), it outputs \(m\) with probability given by the Born rule, \(\Tr(\rho M_{m|y})\).  
    
    \end{itemize}
    Finally, and since transformations can be absorbed either by measurements and/or preparations, the overall probabilities predicted in previous scenarios (often called \gls{PM} scenarios) are given by
    \begin{equation}
    \label{operational stats}
        \Pr[m|x,y]= \Tr(\rho_x M_{m|y}),
    \end{equation}
    where the classical inputs \((x,y)\) unambiguously specify the preparation and measurement for the given protocol setup.

\subsection{Quantum State Discrimination with Post-Measurement Information}
\label{subsec:pi}

    In this section, the formalism of~\cite{PRA:GW10} is introduced, which will be required to analyze the security of the proposed protocols.
    Quantum state discrimination is a specific task in the \gls{PM} scenario. Therein, Bob has no input and tries to decode Alice's classical input with the highest probability by optimally distinguishing between the quantum states which encode her message.
    In~\cite{PRA:GW10}, the state discrimination task is analyzed when classical information related to the preparation is revealed by Alice (who prepares the state according to some information string \(x\) and some encoding \(e\), where the latter is then revealed) to Bob (who measures the state and tries to guess \(x\)).
    But, this reveal is conditioned on the fact that Bob did measure the state and holds no quantum information when receiving this information.
    
    An upper bound is shown to hold when the revealed post-measurement information by Alice (\(e \in \mathcal{E}\) with probability \(p_e\), where \(\mathcal{E}\) is the set of all possible encodings) and the previously measured information by Bob (\(x\) with probability \(p_x\)) form a product distribution (\(p_{x,e} = p_x p_e\)), and for the preparation of the state \(x\) is sampled from the uniform distribution, i.e., \(p_x  = 1/|X|\).
    
    Moreover, without loss of generality, it is assumed that Bob performs a measurement whose outcomes are vectors \({\textbf{m}} = (x^{(1)}, \dots, x^{(|\mathcal{E}|)}) \in X^\mathcal{E}\).
    And depending on the encoding \(e\in \mathcal{E}\) that Bob learns (given to them by Alice) after measuring, Bob will output the guess \(x^{(e)}\).
    
    \begin{lemma}[\cite{PRA:GW10}]\label{lemma:pi}
        Let \(|X|\) be the number of possible strings, and suppose that the joint distribution over strings and encodings satisfies \(p_{x,e} = p_e/|X|\), where the distribution \(\{p_e\}_e\) is arbitrary. Then
        \begin{equation*}
    \Pr_{\mathrm{guess}}^{\mathrm{PI}}[x|E, P] \leq \frac{1}{|X|} \Tr\left[\left(\sum_{\textbf{m}\in X^{\mathcal{E}}} \rho_{\textbf{m}}^\alpha \right)^{\frac{1}{\alpha}}\right]
        \end{equation*}
        for all \(\alpha>1\), where \(E = \{\rho_{x^{(e)},e}\}_{x\in X, e\in \mathcal{E}}\) is the ensemble of all possible states of messages and encodings,     \( P=\{p_{x,e}\}_{x\in X, e\in \mathcal{E}}\) its associated probability distribution
        and \(\rho_{\textbf{m}} = \sum_{e=1}^{{e=|\mathcal{E}|}} p_e\; \rho_{x^{(e)},e}\), the state that corresponds to some outcome vector \(\textbf{m}\).
    \end{lemma}

\subsection{Oblivious Transfer}\label{subsec:ot}
    \textit{\acrlong{OT}} is a protocol between two parties, a \sender{} and a \receiver{}, and can be formulated in different but equivalent functionalities~\cite{C:C87}.
    The most common and perhaps most useful formulation is the \textit{1-out-of-2 \gls{OT}}~\cite{CACM:EGL85}, where two messages are sent by a \sender{} to a \receiver{}, and the \receiver{} is only able to recover one message of its choice with the \sender{} remaining oblivious to which message was received. This intuition is made precise in Definition~\ref{def:ot} by bounding the distance (as given by the trace-norm \(\|A\|_1  =\Tr \sqrt{ \left( A^*A \right)}\)) of the ideal state containing no information useful for cheating, and the actual state produced from a cheating strategy.
    
    \begin{definition}[1-out-of-2 Oblivious Transfer]\label{def:ot}
        A 1-out-of-2 \acrlong{OT} protocol is a protocol between two parties, a \sender{} and a \receiver{}, where the \sender{} has inputs \(x_0,x_1\in\{0,1\}\) and no output, and the \receiver{} has input \(y\in\{0,1\}\) and output \(m\), such that the following properties hold:
        \begin{itemize}
            \item ($\varepsilon$-Correctness) For an honest \sender{} and \receiver{}, \(\Pr[m=x_y\, |\, x_0,x_1,y]\geq 1-\varepsilon\).
            \item ($\varepsilon$-\receiver{}-security): Let \(\rho_{y,x_0,x_1;\tilde{\mathsf{S}}}\) be the state at the end of the protocol with an honest \receiver{} and in the presence of a malicious \sender{}, \(\tilde{\mathsf{S}}\). Then, for all algorithms \(\tilde{\mathsf{S}}\), there exists \((x_0,x_1)\in \{0,1\}^2\), such that \(\Pr[m=x_y]\geq 1-\varepsilon\) and 
            \[
                \left\|\rho_{y,x_0,x_1;\tilde{\mathsf{S}}} - \rho_y \otimes \rho_{x_0,x_1;\tilde{\mathsf{S}}}\right\|_1 \leq \varepsilon.
            \] 
            \item ($\varepsilon$-\sender{}-security) Let \(\rho_{y,x_0,x_1;\tilde{\mathsf{R}}}\) be the state at the end of the protocol with an honest sender and in the presence of a malicious \receiver{}, \(\tilde{\mathsf{R}}\). Then, for all algorithms \(\tilde{\mathsf{R}}\), exists \(y\in \{0,1\}\), such that
            \[
                \left\|\rho_{x_{1-y},x_y,c;\tilde{\mathsf{R}}} -  \frac{\mathds{1}}{2} \otimes \rho_{x_y,y;\tilde{\mathsf{R}}}\right\|_1 \leq \varepsilon.
            \]
        \end{itemize}
        If these properties only hold when restricting the algorithms \(\tilde{\mathsf{S}}\) or \(\tilde{\mathsf{R}}\) to run in probabilistic polynomial time, then the protocol is said to be computationally-secure. 
    \end{definition}

    Despite its simplicity, \gls{OT} is a fundamental primitive in cryptography, and it was shown to be sufficient to construct \gls{MPC}~\cite{STOC:Kilian88}.
    However, no black-box construction of \gls{OT} can exist given only \glspl{OWF} in the classical world~\cite{STOC:IR89}, meaning that \gls{PKC} was compulsory.
    Nevertheless, by also accounting for quantum computation and communication and quantum-secure \glspl{OWF}, \gls{OT} can be achieved without any \gls{PKC} requirement~\cite{C:BCKM21,EC:GLSV21}.
    This means that introducing quantum computation and communication substantially relaxes the requirements to construct \gls{OT}, as candidates for quantum-secure \glspl{OWF} are simpler and more frequent.

\subsection{One-Time Memory and One-Time Programs}
\textit{\acrlong{OTM}} was introduced in~\cite{C:GKR08} as a secure hardware device inspired by the cryptographic protocol of \gls{OT}.
It executes exactly the same functionality as the 1-out-of-2 (chosen-input) \gls{OT} (as presented in Section~\ref{subsec:ot}), defined in~\cite{C:GKR08,TCC:GIS+10}.
However, since it consists of a physical tamper-proof memory with restricted read and write access, it effectively enforces a one-shot \gls{OT} execution, where the \sender{} transmits a single message to the \receiver{} by physically transferring the memory object (the \gls{OTM}).
The reliance of \gls{OTM} on secure hardware presents a significant limitation, as designing tamper-proof hardware is notoriously difficult. Consequently, this requirement has since been relaxed, with cryptographic alternatives introduced in the literature, albeit under certain trusted models~\cite{ISC:DDKZ14,ITCS:Liu23}.

Hence, we choose to formulate the definition of \gls{OTM} in terms of the standard \gls{OT} definition, as it will integrate with our other contributions in this paper (non-interactive \gls{OT}).

\begin{definition}[\acrlong{OTM}]\label{def:otm}
    Let \(\mathcal{M}\) be a 1-out-of-2 \gls{OT} protocol as in Definition~\ref{def:ot}.\\
    If \(\mathcal{M}\) is one-shot, i.e., if the protocol executes with a single message from the \sender{} to the \receiver{}, then \(\mathcal{M}\) is a \acrlong{OTM}.
\end{definition}

An important extension of the \gls{OTM} is the \textit{parallel-\gls{OTM}}, which involves the parallel and non-adaptive execution of multiple \glspl{OTM}.
A parallel-\gls{OTM} consists of multiple independent \glspl{OTM} whose inputs must be chosen concurrently and non-adaptively.
That is, the input to one \gls{OTM} cannot be chosen based on the outcome of any other, as it would in a sequential execution.

\medskip

In~\cite{C:GKR08}, the concept of \textit{\acrlong{OTP}} was also introduced, and subsequently further studied in~\cite{TCC:GIS+10}.
A \gls{OTP} is a program that models a black-box, which computes some function on a given input a single time, preventing any further interaction with the function.
The \gls{OTM} as a concept is particularly interesting as the means to design \glspl{OTP}.

\begin{theorem}[\cite{C:GKR08,TCC:GIS+10} (Informal)]\label{thm:otp-compiler}
    Assuming the existence of one-way functions, there is a polynomial-time compiler that takes any polynomial-time program and a parallel-\gls{OTM}, and returns an \gls{OTP} with the same functionality.
\end{theorem}

\noindent
The original compiler from \gls{OTM} to \gls{OTP} of~\cite{C:GKR08} is heavily based on the technique of Yao's garbled circuits~\cite{FOCS:Yao86} (analogous to using \gls{OT} for \gls{2PC}).
However, Yao's garbled circuits only provide semi-honest security, and in order to guarantee security of the \gls{OTP} against malicious adversaries in a non-interactive manner, \cite{C:GKR08} presents a solution where the \glspl{OTM} output is XORed to mask the output of the garbed circuit.

\subsection{Restricted Quantum-Storage Models}
    Again, it is impossible to achieve unconditional security of \gls{OT} and \gls{BC} without any imposed assumption.
    Therefore, to avoid supporting the security of a protocol on conjectures on computationally-hard problems (e.g., \glspl{OWF} or \gls{PKC}), restrictions to the computation model based on physical phenomenons (motivated by current technology limitations) were introduced for \gls{BC} and \gls{OT}.
    Two main restrictions to the quantum-storage capability of the parties have been introduced.
    First, restrictions to the \textit{storage-space}, either in the dimension of the quantum states that a party can coherently measure~\cite{C:Salvail98}, or on the total size of the storage available~\cite{FOCS:DFSS05}.
    Second, restrictions to the \textit{storage-time} (duration) that a quantum state can be stored before being subjected to quantum decoherence~\cite{PRL:WST08}.

\subsubsection{Bounded-Quantum-Storage Model}
    The \textit{\acrlong{BQSM}}~\cite{FOCS:DFSS05} establishes that there is a point during the protocol, called the \textit{memory bound}, when all but \(M\) qubits of the (otherwise unbounded) memory register of the parties are measured.
    Besides this transient limitation during the execution of the protocol, no restrictions are applied to the classical memory and computing power, which are still considered unbounded.
    
    The functionality of the \gls{BQSM} is described in Definition~\ref{def:bqsm}.
    In this work, it will be assumed that the time instant \(t\) and memory size \(M\) of Definition~\ref{def:bqsm} are set in advance, when designing a protocol in the \gls{BQSM}.
    \begin{definition}[Bounded-Quantum-Storage Model]\label{def:bqsm}
        The \acrlong{BQSM} consists of two identically modeled computation phases \(\mathcal{P}_\mathrm{pre}, \mathcal{P}_\mathrm{post}\), discontinued by a partial measurement of the memory register of the parties \(\mathcal{M}_{t,M}\), where (in chronological order):
        \begin{enumerate}
            \item \(\mathcal{P}_\mathrm{pre}\): the state of a party may have an arbitrary number of qubits (\(N\)), and arbitrary computations are allowed over this system.
            \item \(\mathcal{M}_{t, M}\): at a certain point in time \(t\), the memory bound applies, i.e., all but \(M\leq N\) qubits are measured.
            \item \(\mathcal{P}_\mathrm{post}\): the party is again unbounded in quantum memory and computing power.
        \end{enumerate}
    \end{definition}

\subsubsection{Noisy-Quantum-Storage Model}
    Generalizing the \gls{BQSM} to a more realistic noisy-memory model is left as an open question in \cite{FOCS:DFSS05}.
    The \textit{\acrlong{NQSM}}~\cite{PRL:WST08,TIT:KWW12} addresses this weaker assumption by considering the quantum memory of the parties performing the protocol to be imperfect due to the presence of noise.
    This model represents a more realistic setting given the current available technology, and does not require an arbitrary estimation of the total memory available to an all-powerful adversary.
    In opposition, any qubit that is stored experiences noise that leads to quantum decoherence.
    
    The functionality of the \gls{NQSM} is given in Definition~\ref{def:nqsm}.
    Again, in this work, it will be assumed that the family \(\{\mathcal{F}_t\}\) of Definition~\ref{def:nqsm} is known in advance when designing a protocol in the \gls{NQSM}.
    Note that the \gls{BQSM} is a particular case of the \gls{NQSM}, where \(\mathcal{F}_t= \mathds{1}\) for all \(t\) but the dimension of \(\mathcal{H}_\mathrm{in}\) is bounded.
    
    \begin{definition}[Noisy-Quantum-Storage Model]\label{def:nqsm}
        Let \(\rho \in \mathcal{L}(\mathcal{H}_\mathrm{in})\) be a quantum state stored in a quantum memory.
        The \acrlong{NQSM} prescribes a family of completely positive trace-preserving functions \(\{\mathcal{F}_t\}_{t \geq 0}\), such that the content of the memory after a certain time \(t\) is a state \(\mathcal{F}_t(\rho)\), where \(\mathcal{F}_t : \mathcal{L}(\mathcal{H}_\mathrm{in}) \to \mathcal{L}(\mathcal{H}_\mathrm{out})\), and 
        \[\mathcal{F}_0 = \mathds{1}\quad \mathrm{and}\quad \mathcal{F}_{t_1 + t_2} = \mathcal{F}_{t_1} \circ \mathcal{F}_{t_2},\]
        i.e., noise in storage only increases with time.
    \end{definition}
    
    To enable an analysis of the relation between the storage size and the probability of successfully decoding stored states, it is often considered that the memory is composed by \(N\) different cells and that noise affects these cells separately, i.e., \(\mathcal{F} = \mathcal{N}^{\otimes N}\).
    Then for a large enough \(N\), the probability that a party can decode some rate \(R\) (above the classical capacity of the channel, \(C_\mathcal{N}\)) of its quantum memory decays exponentially with \(N\)~\cite{TIT:KWW12}:
    \begin{align}\label{eq:nqsm}
        \Pr_{\mathrm{succ}}^{\mathcal{N}^{\otimes N}}[N R] & \leq 2^{N\cdot \gamma^\mathcal{N}(R)}, \\
        \gamma^{\mathcal{N}}(R) &> 0 \quad \text{for all} \quad R>C_\mathcal{N}.\nonumber
    \end{align}
    An example of noisy channel is the \(d\)-dimension depolarizing channel \(\mathcal{N}_r:\mathcal{L}(\mathcal{H}) \to \mathcal{L}(\mathcal{H})\), for \(d\geq 2, 0\leq r\leq1\),
    \begin{equation}\label{eq:depolarizingchannel}
        \mathcal{N}_r (\rho) \to r\rho + (1-r)\frac{\mathds{1}}{d},
    \end{equation}
    which gradually converts a stored state \(\rho\) to a maximally mixed state with probability \(1-r\).
    
    Note that the assumption that \(\mathcal{F} = \mathcal{N}^{\otimes N}\) considers storing each qubit independently.
    This means that even if two qubits are entangled, the entanglement is not affected by more than the independent noise that each qubit undergoes by itself, which still leads to the degradation of the entanglement.

\section{Non-Interactive OT}\label{sec:ITProtocol}
    In this section, our first contribution of an unconditionally-secure two-message non-interactive 1-out-of-2 \gls{OT} is presented.
    A construction is given in the \gls{NQSM} and its security is proved in this model.
    Also, an alternative construction for a non-interactive 1-out-of-2 \gls{OT} is presented, which removes the time-delay constraint of the previous \gls{NQSM} construction in exchange for adopting the \gls{BQSM}.

    Regarding the \gls{NQSM} (Definition~\ref{def:nqsm}), we will make a simplification by parameterizing our protocol by a time bound \(\tau\) that enforces total decoherence of the memories of the parties.
    This model may, for instance, be interpreted as a depolarizing channel (Equation~\eqref{eq:depolarizingchannel}) that after \(\tau\) time steps erases all information about state \(\rho\), i.e., \({(\mathcal{N}_r)}^\tau (\rho) = \mathds{1}/d\).
    One could instead study different noise models and the dependence of the security of the protocol with the noise level at any point in time \(t<\tau\).
    We explicitly choose to parameterize our protocol directly by the time to total decoherence \(\tau\), as it represents the worst case scenario for an adversary.
    Also, this closely relates the \gls{BQSM} as the limit of the \gls{NQSM}.

\subsection{Preliminaries}
\label{subsec:Prelims}

We first introduce some basic definitions and  notation for key elements of the protocol. Let \([N] := \{n \; |\; n \in \mathbb{N} \;\text{and}\; n \leq N\}\), and \(x = x_0\; x_1 \in X = \{00, 01, 10, 11\}\) be  the \textit{message}.

\begin{definition}[\(N\)-qubit register]
    We refer to a set of \(N\) qubits, \(\mathsf{R} = \{\mathsf{q}_1, \dots, \mathsf{q}_i, \dots, \mathsf{q}_N\}\), indexed by \(i \in [N]\),  as an \textit{\(N\)-qubit register}. An element, \(\mathsf{q}_i\), of the register is interpreted as the physical system at the \(i\)-th site, rather than the operational description of its quantum system.
\end{definition}

\begin{definition}[Index-encoding set]
    Let the \textit{index-encoding set} be a set of tuples \(\mathcal{E} := \{(k,\ell) \;|\; k, l \in [N] \;\textrm{and}\; k < l \}\), where \(|\mathcal{E}| = \binom{N}{2} = N(N-1)/2\). 
    Then, the set \(\mathcal{E}\) is the set of ordered tuples \((k,\ell)\) where \(k < l\), such that an element of the index-encoding set selects a pair of distinct sites \((\mathsf{q}_k, \mathsf{q}_\ell)\) of the register.
\end{definition}

\begin{definition}[Sub-Register]
    Let \(\mathsf{R}\,\backslash\,{\{\mathsf{q}_{i_1},\dots, \mathsf{q}_{i_n}\}}\) be an \textit{\((N-n)\)-qubit sub-register} of \(\mathsf{R}\), indexed by \([N]\backslash{\{i_1,\dots,i_n\}}\). We write \(\rho_{[N]\backslash{\{i_1,\dots,i_n\}}}:= \rho_1 \otimes \cdots \otimes \rho_{i_1-1} \otimes \rho_{i_1+1} \otimes \cdots \otimes \rho_{i_n-1} \otimes \rho_{i_n+1} \otimes \cdots \otimes \rho_N\), to denote that the quantum state in each site \(j\) of the sub-register is equal to \(\rho\), i.e., \(\rho_j = \rho\) for all \(j\).
\end{definition}

\begin{definition}[Message encoding vector]\label{def:encoding}
    Given the set of all possible assignments from the index-encodings to the messages \(X^{\mathcal{E}} = \left\{ \textbf{m}_1, \dots, \textbf{m}_{4^{|\mathcal{E}|}}\right\}\), let the \textit{message encoding vector} be the specific assignment \(\textbf{m}_i \), which is explicitly  denoted as \(\textbf{m}_i =\left(\langle x_0^{(k)}\, x_1^{(\ell)}\rangle_i\;|\; (k,\ell) \in \mathcal{E} \right)\). 
\end{definition}

Note that there are \(4^{|\mathcal{E}|}\) possible message encoding vectors, and each vector \(\textbf{m}_i\) has \(|\mathcal{E}|\) entries. One may assume that the index \(i\) gives the placement of the vector in lexicographical order, for example, \(\textbf{m}_1 =\left(\langle 0^{(k)}\, 0^{(\ell)}\rangle_1\;|\; (k,\ell) \in \mathcal{E} \right)\) and \(\textbf{m}_{4^{|\mathcal{E}|}} =\left(\langle 1^{(k)}\, 1^{(\ell)}\rangle_1\;|\; (k,\ell) \in \mathcal{E} \right)\).
While the previous definition assumes a level of generality where the message content could be correlated with the index-encoding, this is not something we consider in the proposed protocol.
We assume that the index-encodings \((k,\ell)\) are randomly sampled and independent of the chosen message \(x_0,x_1\). Nevertheless, we adopt this level of generality as it will be required when proving security, namely, when using the discrimination framework with post-measurement classical information of~\cite{PRA:GW10} (see Section \ref{subsec:pi}).

\begin{definition}[Message encoding state]\label{def:mes}
    Let \(\textbf{m}_i =\left(\langle x_0^{(k)}\, x_1^{(\ell)}\rangle_i\;|\, (k,\ell) \in \mathcal{E} \right)\) be a message encoding vector as in Definition~\ref{def:encoding}, then, its associate \textit{message encoding state} is given by
     \begin{equation*}
    \rho_{\textbf{m}_i} =  \rho_{\langle x_0^{(k)} x_1^{(\ell)}\rangle_i} = \frac{1}{ |\mathcal{E} |\cdot 2^{N-2}}\left(\sum_{k<\ell}\ketbra{B_{\langle x_0^{(k)} x_1^{(\ell)}\rangle_i}}_{k,\ell}\otimes \mathds{1}_{[N]\backslash {\{k,\ell\}}}\right),
    \end{equation*} 
    which describes the density matrix for the \(N\)-qubit register \(\mathsf{R}\) in full generality, allowing the message to depend on the uniformly sample index-encodings \(\mathcal{E}\).  
\end{definition}

It will also be useful to consider the unnormalized version of the state \(\sigma_{\mathbf{m}_i} = \rho_{\mathbf{m}_i} \cdot {|\mathcal{E}|\cdot2^{N-2}}\).

\begin{lemma}\label{lemma:weylineq}
   Let \( A, B \in M_n \)  be Hermitian  positive semi-definite matrices.  Then, \(
\lambda_{\textup{max}}(A+B) \leq \lambda_{\textup{max}}(A) + \lambda_{\textup{max}}(B).
\)
\end{lemma}

\begin{proof}
 The spectral norm of a Hermitian matrix \( M \), denoted \( \|M\|_2 \), is equal to the largest eigenvalue in magnitude, i.e., \(\|M\|_2 = \max_i \{|\lambda_i|\},\) where \( \lambda_i \) are the eigenvalues of \(M\). Since \(A\) and \(B\) are also positive semi-definite, all their eigenvalues are non-negative. Therefore, the spectral norm of \(A\) and \(B\) becomes \(
\|A\|_2 = \lambda_{\max}(A), \; \|B\|_2 = \lambda_{\max}(B),
\) respectively. 

The triangle inequality for the spectral norm states that \(
\|A + B\|_2 \leq \|A\|_2 + \|B\|_2.
\) Since \(A+B \) is also Hermitian and positive semi-definite we have \(
\|A + B\|_2 = \lambda_{\max}(A + B)
\) and by direct substitution we get
\(
\lambda_{\max}(A + B) \leq \lambda_{\max}(A) + \lambda_{\max}(B).
\)
\end{proof}

Finally, in Lemma~\ref{lemma:lambdasigstar}, we introduce an important lemma giving a maximal eigenvalue upper bound, which will be essential for the security proof.
\begin{lemma}\label{lemma:lambdasigstar}
    Let \(\mathbf{m}_i =\left(\langle x_0^{(k)}\, x_1^{(\ell)}\rangle_i\;|\;(k,\ell) \in \mathcal{E}\right)\) be a message encoding with unnormalized associated message encoding state
    \[\sigma_{\mathbf{m}_i}=\left(\sum_{k<\ell}\ketbra{B_{\langle x_0^{(k)} x_1^{(\ell)}\rangle_i}}_{k,\ell}\otimes \mathds{1}_{[N]\backslash {\{k,\ell\}}}\right).\]
    Then, the largest eigenvalue, \(\lambda_{\textup{max}}(\sigma_{\textbf{m}_{i}})\), is upper bounded by
   \[\lambda_{\textup{max}}(\sigma_{\mathbf{m}_{i}})\leq \frac{N^2}{4} +\frac{N}{4}-\frac{1}{2}.\]
\end{lemma}
\begin{proof}
Let us start by defining a shorthand notation, where we also make explicit the terms \(\langle x_0^{(k)} \, x_1^{(\ell)}\rangle_i\) of the message encoding in the state and the size of the register \(N\),
\begin{equation}
\label{eq:sigmastate}
    \sigma_{N}^{\langle x_0^{(k)} x_1^{(\ell)}\rangle_i} = \sum_{k<\ell} \mathbb{B}_{\langle x_0^{(k)} x_1^{(\ell)}\rangle_i}
\end{equation}
with  
\begin{equation}
    \mathbb{B}_{\langle x_0^{(k)} x_1^{(\ell)}\rangle_i} =\ketbra{B_{\langle x_0^{(k)} x_1^{(\ell)}\rangle_i}}_{k,\ell}\otimes \mathds{1}_{[N]\backslash {\{k,\ell\}}}.
\end{equation}
For an \(N\)-qubit register, the previous state \( \sigma_{N}^{\langle x_0^{(k)} x_1^{(\ell)}\rangle_i} \) can be interpreted as a sum over the \(   |\mathcal{E}|=\binom{N}{2}\) edges of the complete graph \(K_N\), where each vertex represents a qubit and each edge connects qubits \(k\) and \(\ell\), and is given by state \(\mathbb{B}_{\langle x_0^{(k)} x_1^{(\ell)}\rangle_i}\) for a specific message encoding \(\mathbf{m}_i\). Noticing this, we can rewrite Equation \eqref{eq:sigmastate} by separating the summation domain over the edges into two disjoint subsets as follows
\begin{equation}\label{eq:CGplusSG}
    \sigma_{N}^{\langle x_0^{(k)} x_1^{(\ell)}\rangle_i} = \sigma_{N-1}^{\langle x_0^{(k)} x_1^{(\ell)}\rangle_i} + \sigma_{star(N)}^{\langle x_0^{(k)} x_1^{(\ell)}\rangle_i},
\end{equation} where 
\begin{equation}\label{eq:starstatee}
    \sigma_{star(N)}^{\langle x_0^{(k)} x_1^{(\ell)}\rangle_i} = \sum_{j=1}^{N-1}\mathbb{B}_{\langle x_0^{(j)} x_1^{(N)}\rangle_i}
\end{equation}
is the unnormalized mixture of all Bell pairs involving the $N$th qubit. Using the graph interpretation described above, such state can be seen as a star graph with its center at the $N$th vertex, the latter being connected to all other $N-1$ vertices.  
This relation can be applied recursively,  allowing the expression of the \( \sigma_{N}^{\langle x_0^{(k)} x_1^{(\ell)}\rangle_i} \) state as a sum of  \(\sigma_{star(n)}^{\langle x_0^{(k)} x_1^{(\ell)}\rangle_i}\) states, for $n\in \{2, .., N\}$.

Since the states in Equation~\eqref{eq:CGplusSG} correspond to Hermitian positive semi-definite matrices, applying Lemma \ref{lemma:weylineq} we get the following upper bound for the maximum eigenvalue, 
\begin{equation}
\label{recursion}
     \lambda_{\textup{max}}\left(\sigma_{N}^{\langle x_0^{(k)} x_1^{(\ell)}\rangle_i}\right) \leq  \lambda_{\textup{max}} \left(\sigma_{N-1}^{\langle x_0^{(k)} x_1^{(\ell)}\rangle_i}\right) +  \lambda_{\textup{max}}\left(\sigma_{star(N)}^{\langle x_0^{(k)} x_1^{(\ell)}\rangle_i}\right).
\end{equation} 

Next, notice that we can apply local unitary transformations at each \(j\) qubit to transform it into any Bell pair of our choosing, and since the spectrum is invariant under unitary transformations we have that, for all \(i\),
\begin{equation}
    \lambda_{\textup{max}}\left(\sigma_{star(N)}^{\langle x_0^{(k)} x_1^{(\ell)}\rangle_i}\right) = \lambda_{\textup{max}}(\sigma_{star(N)}).
\end{equation}
Without loss of generality, let us consider \(\ketbra{B_{11}}\), obtained by applying \(\big(Z^{{x_0^{(j)} \oplus 1}}X^{{x_1^{(j)}}\oplus 1}\big)_j \otimes \mathds{1}_c \) to \(\mathbb{B}_{\langle x_0^{(j)} x_1^{(c)}\rangle}\). Thus, with the foresight that our attention will lie only in the spectrum of the operators, we can write
\begin{equation}
   \sigma_{star(N)} = \sum_{j=1}^{N-1}\ketbra{B_{ 1 1}}_{j,c}\otimes \mathds{1}_{[N]\backslash {\{j,c\}}}.
\end{equation}

Rewriting the Bell state in terms of the Pauli matrices (Equation~\eqref{Pauli}) we have
    \begin{align}
    \begin{split}
         \sigma_{star(N)} &= \frac{1}{4}\sum_{j=1}^{N-1}\bigg(\mathds{1}_j \otimes \mathds{1}_c - X_j \otimes X_c- Z_j \otimes Z_c - Y_j \otimes Y_c\bigg) \otimes \mathds{1}_{[N]\backslash \{j,c\}} \\ 
         &= \frac{N-1}{4}\mathds{1}_{[N]}- \frac{1}{4}\left(\sum_{j=1}^{N-1} \textbf{X}_j \cdot \textbf{X}_c+ \textbf{Z}_j \cdot \textbf{Z}_c+ \textbf{Y}_j \cdot \textbf{Y}_c\right).
    \end{split}
    \end{align}
where  \( \textbf{X}_i = \mathds{1}_{1}\otimes \ldots \otimes  {X}_i \otimes \ldots\otimes \mathds{1}_{N}\) such that \(\textbf{X}_j \cdot \textbf{X}_c = X_j \otimes X_c \otimes \mathds{1}_{[N]\backslash\{j,c\}}\), and similarly for \(\textbf{Y}_i\) and \(\textbf{Z}_i\).

Finally, let us rewrite the previous expression as
\begin{equation}
   \sigma_{star(N)}= \frac{N-1}{4}\mathds{1}_{[N]} - \mathbf{H}^{(N)}_{star}, 
\end{equation} where \begin{equation} \mathbf{H}^{(N)}_{star} = \frac{1}{4}\left(\sum_{j=1}^{N-1} \textbf{X}_j \cdot \textbf{X}_c+ \textbf{Z}_j \cdot \textbf{Z}_c+ \textbf{Y}_j \cdot \textbf{Y}_c\right)\end{equation} is known as the Heisenberg-star spin model in many-body physics \cite{RICHTER19951497}.
Focusing on the largest eigenvalue for \(\sigma_{star(N)}\), we have the following relation, 
\begin{align}\begin{split}\label{eq:ineqss}
    \lambda_{\text{max}}\left(\sigma_{star(N)}\right) &=\frac{N-1}{4} + \lambda_{\text{max}}\left(-\mathbf{H}^{(N)}_{star}\right)\\
    &= \frac{N-1}{4} - \lambda_{\text{min}}\left(\mathbf{H}^{(N)}_{star}\right),
\end{split}\end{align}
where we have rewritten the equation in terms of the minimum eigenvalue for \(\mathbf{H}^{(N)}_{star}\), which corresponds to the ground-state energy of the Heisenberg-star spin system, calculated analytically in \cite{RICHTER19951497} to be 
\begin{equation}\label{eq:lambdamax}
    \lambda_{\text{min}}\left(\mathbf{H}^{(N)}_{star}\right) = -\frac{1+N}{4}.
\end{equation}
From Equations~\eqref{eq:ineqss} and~\eqref{eq:lambdamax}, we obtain that
\begin{equation}
\lambda_{\text{max}}\left(\sigma_{star(N)}\right)= \frac{N}{2}.
\end{equation}
Finally, taking Equation~\eqref{recursion} and using it recursively (until there is only one Bell pair left), we achieve the desired result
\begin{align}\begin{split}
     \lambda_{\textup{max}}\left(\sigma_{N}^{\langle x_0^{(k)} x_1^{(\ell)}\rangle_i}\right) 
     &\leq  \sum_{n=2}^{N} \lambda_{\text{max}}\left(\sigma_{star(n)}\right)\\
     &\leq \sum_{n=2}^{N} \frac{n}{2} \\
     &\leq  \frac{N^2}{4} +  \frac{N}{4}-\frac{1}{2}.
\end{split}\end{align}
\end{proof}

\subsection{Non-Interactive OT Protocol}\label{subsec:NIOT}

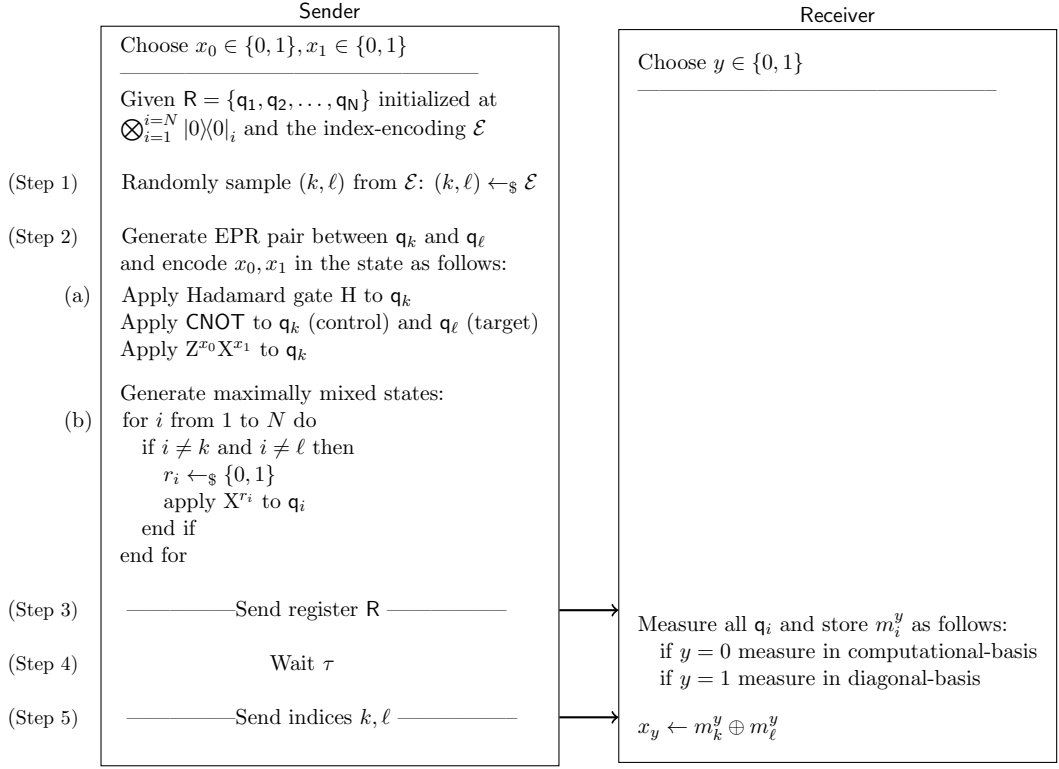
\begin{figure}[htb]
    \centering
    \input{protocolnewX}
    \caption{Schematic representation of the proposed two-message non-interactive \gls{OT} protocol parameterized by \(N(\sigma), \tau\). The ``\textsf{Wait} \(\tau\)'' procedure by the \sender{} may be disregarded in exchange for a larger \(N\) (Section~\ref{subsec:noDt}).}
    \label{fig:protocol}
\end{figure}

    Intuitively, to implement the \gls{OT} protocol, the \sender{} will hide an \gls{EPR}-pair encoding its two bits \(x_0,x_1\), masked among many ``decoy'' qubits of the \(N\)-qubit register  \(\mathsf{R} = \{\mathsf{q}_1, \dots, \mathsf{q}_N\}\), such that the \receiver{} cannot know which qubits are encoding the information without the \sender{} revealing them. A detailed operational description of the protocol is given in Figure \ref{fig:protocol}. Furthermore, an informational perspective from the view of the \sender{} and \receiver{} is introduced below.
    \begin{itemize}[itemsep=2pt]
        \item \textbf{Step 0}: The \sender{} chooses \(x_0 \in \{0,1\}\), \(x_1 \in \{0,1\}\) and sets up an  \(N\)-sized qubit register  \(\mathsf{R} = \{\mathsf{q}_1, \dots, \mathsf{q}_N\}\), where \(N\) depends on the security parameter \(\sigma\),  initialized in the state \(\bigotimes_{i=1}^{i=N}\ketbra{0}{0}_i\). The \receiver{} chooses \(y \in \{0,1\}\).
        \item \textbf{Step 1}: 
        The \sender{} uniformly samples indices \(k,\ell\) from the index encoding set \(\mathcal{E}\) (with \(k<\ell\), without loss of generality), selecting qubits \(\{\mathsf{q}_k,\mathsf{q}_\ell\} \subset \mathsf{R}\). 
        \item \textbf{Step 2}: 
        \begin{itemize}
            \item[(a)] The \sender{} maximally  entangles qubits \(\{\mathsf{q}_k,\mathsf{q}_\ell\}\), \( \ket{B_{00}}_{k,\ell} = \left(\mathsf{CNOT})_{k,\ell} \cdot (H_k \otimes \mathds{1}_\ell \right) \ket{00}_{k,\ell}\).
            Furthermore, it encodes \(x_0,x_1\) in the entangled pair of  qubits \(\mathsf{q}_k,\mathsf{q}_\ell\) accordingly, \( \ket{B_{x_0 x_1}}_{k,\ell} = \left((\mathrm{Z}^{x_0} \mathrm{X}^{x_1})_k \otimes \mathds{1}_\ell \right) \ket{B_{00}}_{k,\ell},\) leading to the state 
            \begin{equation*}
                \ketbra{B_{ x_0x_1}}_{k,\ell}\otimes \ketbra{0}_{[N]\backslash {\{k,\ell\}}}.
            \end{equation*}
            \item[(b)] The \sender{} generates maximally-mixed states for the remainder of the register, by implementing 
            \(\mathrm{X}^{r_i}\) for random bit \(r_i\) to \(\ketbra{0}_{i}\) for \(i \in [N]\backslash{\{k,\ell\}}\)\begin{equation*}\ketbra{B_{ x_0x_1}}_{k,\ell}\otimes \frac{1}{2^{N-2}}\mathds{1}_{[N]\backslash {\{k,\ell\}}}.\end{equation*}
        \end{itemize}
        
        \item \textbf{Step 3}: The \sender{} sends the entire register \(\mathsf{R}\) to the \receiver{}. For each of the four possible \(x_0,x_1\) choices there is a corresponding state 
        \begin{equation}
            \rho_{\langle x_0,x_1 \rangle} 
            = \frac{1}{|\mathcal{E}|} \frac{1}{2^{N-2}}\sum_{ k < \ell }  \ketbra{B_{ x_0x_1}}_{k,\ell}\otimes \mathds{1}_{[N]\backslash {\{k,\ell\}}}.
        \end{equation} 
        Notice that the previous states correspond to the message encoding state  (Definition \ref{def:mes}) for each of the four constant message-encoding vectors. Indeed, the \sender{} will choose the message independently of the particular index-encoded sampled.
        
        \item \textbf{Step 4}: The \sender{} waits for a pre-determined time \(\tau\), specified by the \gls{NQSM}, for the memory to completely decohere.
        The \receiver{} measures each individual qubit \(\mathsf{q}_i\), either in  computational basis if \(y=0\) or in the diagonal basis \(y=1\),  and stores all classical measurement results \(m^y_i\in \{0,1\}\). 
        
        \item \textbf{Step 5}: Finally, the \sender{} sends the encoding indices \(k,\ell\), to the \receiver{}. The \receiver{}  computes the parity of the stored measurement outputs for \(\{\mathsf{q}_k, \mathsf{q}_\ell\}\), that is, \( m^y_k\oplus m^y_\ell = x_y\).
    \end{itemize}

    In this protocol, the honest \receiver{} will measure individually each qubit in the register, for which no quantum memory is needed.
    As such, a necessary aspect for the security is that the \receiver{} be forced to measure the qubits separately, otherwise, a straightforward attack is to perform \gls{SDC}~\cite{PRL:BW92} and recover both the inputs of the \sender{}.
    One way to mitigate this, as we did, is by imposing the constraints offered by the \gls{NQSM}, wherein the \sender{} will need to wait a fixed amount of time (\(\tau\)) in order for the memory of any malicious \receiver{} to decohere.  Therefore, either the \receiver{} proceeds honestly and according to the protocol prescription measures every qubit separately,  or it acts maliciously and tries to implement a general measurement over the register before losing the encoded message to decoherence. The size of the register \(N\) must be set to ensure unconditional security, which defines the success of a malicious actor in a statistical experiment of running the protocol.
    We will show that the success probability of any possible attack (i.e., the statistical distance between distributions) goes to zero linearly with \(N\).
    Thus, we will set \(N\) to be a large enough constant (independent of the adversarial power), such that the statistical distance between distributions is negligible.
        
    Regarding the waiting time \(\tau\), we remark that one instance of waiting \(\tau\) can be ``reused'' for many parallel executions of the protocol. And since \gls{OT} is often used as a building block for other primitives, and these often require many \gls{OT} executions, this delay can be amortized over all the parallel processes. Nevertheless, in some scenarios it could be perceived as undesirable the need to have an explicit time delay embedded in the design of the protocol, specially when such a delay is substantial when comparing with the generating and transmission of the required messages (qubits and indices) that can be as fast as the speed of light.
    As an alternative, one can remove the delay without affecting the unconditional security, by changing the \gls{NQSM} with the \gls{BQSM}.  We analyze this approach in more detail in Section~\ref{subsec:noDt}, where the \sender{} does not wait any time but the number of qubits that it sends (\(N\)) before revealing the indices \(k,\ell\) is chosen to be large enough, such that the \receiver{} cannot store all of them (from the \gls{BQSM} assumption). Thus, it must guess which subset of \(M\) qubits to store.

\subsubsection{Correctness and Security}\label{subsec:security}
    To establish that the protocol of Figure \ref{fig:protocol} implements a secure 1-out-of-2 \gls{OT}, it must be proved, according to the requirements of Definition~\ref{def:ot}, that: the honest execution of the protocol is correct; the \sender{} does not acquire any information regarding the input of the \receiver{}; and, the \receiver{} remains oblivious to the input of the \sender{} that was not retrieved.
    
    To accomplish such requirements, start by noticing that all communication in the protocol flows from the \sender{} to the \receiver{}, i.e., it is a non-interactive protocol.
    So first, the \sender{} must not be able to keep any (arbitrary-dimension) entangled system with the system it sends to the \receiver{} that would allow the \sender{} to somehow gather any information about the input of the \receiver{} later.
    And second, the \receiver{} must not be able to design any arbitrary-dimension \gls{POVM} over the \(N\)-dimensional state it received from the \sender{} that would allow the \receiver{} to extract more information than one of the messages of the \sender{}.
    These two properties will be formally proved below, but intuitively, they follow from the inability to extract information from future events for the first case, and from combining the \gls{NQSM} (by introducing long-enough delay that imposes decoherence of memories) with the hiding of the qubits encoding for the second case.
    
    The \gls{OT} protocol is parameterized by the statistical security parameter \(\sigma\), and by the time \(\tau\) to quantum decoherence of memories (up to an exponentially low probability \(2^{-\sigma}\)) predefined by the \gls{NQSM} where the protocol is resolved.

\paragraph{Remark.} 
    We remark that the number of transmitted qubits, $N$, must be set as $N=2^\sigma$, meaning that the communication depends on the statistical security parameter, for the statistical distance to be at least $2^{-\sigma}$.
    However, the circuit to prepare each of the states is constant-size and no memory is required.
    Moreover, as this is a \textit{statistical security parameter} that enforces the indistinguishability between two distributions in a single experiment, it is fixed and does not scale with the power of an adversary (that may even be all-powerful).
    This contrasts with \textit{computational security}, where the advantage of the adversary must go to zero faster than any polynomial, because an adversary is allowed polynomial-many tries to distinguish two distributions.
    Indeed, our parameter $N$ is a \textit{constant} in the protocol parametrization for any desired statistical distance, being independent of any computational security parameter, i.e., does not grow with the adversarial power.
        
\begin{theorem}\label{thm:security1}
    The protocol from Figure~\ref{fig:protocol} implements a 1-out-of-2 \acrlong{OT} protocol secure against computationally unbounded adversaries (unconditional security parameterized by \(\sigma\)) in the \acrlong{NQSM} with time to total decoherence \(\tau\).
\end{theorem}
\begin{proof}
    The protocol from Figure~\ref{fig:protocol} is a two-party protocol, where the \sender{} has two inputs \((x_0,x_1)\) and the \receiver{} has one input \(y\) and outputs \(x_y\), which performs precisely the functionality of \gls{OT} (Definition~\ref{def:ot}).
    We will now show the three necessary properties of correctness, \receiver{}-security and \sender{}-security.
    
\paragraph{Correctness:}
   The correctness of the honest strategy for the protocol can be immediately established since it will correspond to a ``stochastic dense coding'' \cite{PRL:PPZCT22} applied to qubits \(\{\mathsf{q}_k, \mathsf{q}_\ell\}\). Therein, both bits are encoded into the Bell state, namely, \(x_0\) is encoded in the phase, and  \(x_1\)  in the parity of the Bell state (just as in \gls{SDC}), but only one bit may be deterministically extracted when using separable measurements. Accordingly, the \receiver{} can either extract the first or second bit by measuring, respectively, the phase  or the parity observables.
   That is, measuring in the computational or the diagonal basis individually for all qubits of the register \(\mathsf{R}\), and deterministically extract the desired bit out of the Bell state shared between \(\{\mathsf{q}_k, \mathsf{q}_\ell\}\) by computing the parity of the individual measurement outputs after receiving the indices. This shows the protocol to have perfect correctness, since an honest strategy will deterministically return \(x_y\).
   We further remark that no quantum memory is required to correctly execute the protocol, and thus, no analysis of the \gls{NQSM} is required.
    
\paragraph{\receiver{}-security:}
    To prove that the protocol is secure for an honest \receiver{}, i.e., against a malicious \sender{}, it must be guaranteed that no matter what the \sender{} does, it cannot recover the input of the \receiver{} (\(y\)).
    In this case, it must be noted that the \receiver{} exclusively performs measurements on its part of the system, and does not explicitly communicate anything to the \sender{}, i.e, communication is one-way.
    Thus, any correlated event that the \sender{} can exhibit (\(Z\)) must be constrained by the \textit{no-signalling from the future}~\cite{PRA:CDP09} (also called \textit{no-backward-in-time signaling}~\cite{Quantum:GSSSBP19}), i.e.,
    \begin{equation}
        \Pr\left[Z|X=x_0 x_1,Y=y\right] = \Pr[Z|X=x_0 x_1].
    \end{equation}
    This, in turn, implies that any correlation that the \sender{} holds (\(Z\)) in its state (\(\rho_{y,x_0,x_1;\tilde{\mathsf{S}}}\)) is conditionally independent of the input of the \receiver{} (\(y\)), \(\rho_{y,x_0,x_1;\tilde{\mathsf{S}}} = \rho_y \otimes \rho_{x_0,x_1;\tilde{\mathsf{S}}}\), as required by Definition~\ref{def:ot}. 
    So,
    \begin{equation}                
        \left\|\rho_{y,x_0,x_1;\tilde{\mathsf{S}}} - \rho_y \otimes \rho_{x_0,x_1;\tilde{\mathsf{S}}}\right\|_1 =0.
    \end{equation}
    Therefore, the \sender{} cannot obtain any information about the input of the \receiver{}, meaning that the protocol has perfect security, in this case.
    
\paragraph{\sender{}-security:}
    To prove that the protocol is secure for an honest \sender{}, i.e., against a malicious \receiver{}, it must be unfeasible for the \receiver{} to recover more than one of the messages of the \sender{}.
    For this, the proof will require enforcing the \receiver{} to measure before receiving the encoding, and then using the formalism of post-measurement information (Section~\ref{subsec:pi}) to analyze the implications (or lack thereof) of sending the encoding.
    
    \medskip
    
    \noindent
    Recall that for each message encoding vector \(\textbf{m}_i\) (Definition~\ref{def:encoding}) there is a corresponding message encoding state \(\rho_{\textbf{m}_i}\) (Definition~\ref{def:mes}), 
    \begin{equation}
        \rho_{\textbf{m}_i} = \frac{1}{ |\mathcal{E}| \cdot 2^{N-2}}\left(\sum_{k<\ell}\ketbra{B_{\langle x_0^{(k)} x_1^{(\ell)}\rangle_i}}_{k,\ell}\otimes \mathds{1}_{[N]\backslash {\{k,\ell\}}}\right).
    \end{equation}
    Now, we consider the post-measurement information formalism introduced in Section~\ref{subsec:pi}, and from Lemma \ref{lemma:pi}  we have that if  $p_{x,e} = p_e/{|X|}$, then
    \begin{equation}
        \Pr^{PI}_{\text{guess}}(x|\mathcal{R}) \leq \frac{1}{|X|} \Tr\left[
      \left(  \sum_{\textbf{m}_i\in X^{\mathcal{E}} }^{} \rho_ {\textbf{m}_i}^{\alpha}\right)^{1/\alpha}
    \right], 
    \end{equation}
    for any \(\alpha > 1\). Thus, applied to our scenario where \(|X|=4\) and \(p_{x,e} = \frac{1}{4}\frac{1}{|\mathcal{E}|}\), then
    \begin{equation}\label{eq:prpiiaaa}
        \Pr^{PI}_{\text{guess}}(x|\mathcal{R}) \leq \mathcal{I}_{\alpha}(N)
    \end{equation}
    for 
    \begin{align}\begin{split}\label{eq:iaaa}
        \mathcal{I}_{\alpha} (N) &:=  \frac{1}{4} \Tr\left[
        \left(  \sum_{\textbf{m}_i\in X^{\mathcal{E}} }^{} \rho_ {\textbf{m}_i}^{\alpha}\right)^{1/\alpha}
        \right] \\
        &=  \frac{1}{ |\mathcal{E}| \cdot 2^{N}}\Tr\left[
        \left(  \sum_{\textbf{m}_i\in X^{\mathcal{E}} }^{} \sigma_{\textbf{m}_i}^{\alpha}\right)^{1/\alpha}
        \right],
    \end{split}\end{align}
    where  
    \begin{equation}\label{eq:sigma}
        \sigma_{\textbf{m}_i} = \left(\sum_{k<\ell}\ketbra{B_{\langle x_0^{(k)} x_1^{(\ell)}\rangle_i}}_{k,\ell}\otimes \mathds{1}_{[N]\backslash {\{k,\ell\}}}\right).
    \end{equation}
    Let  
    \(\Tr\left[\left(  \sum_{\textbf{m}_i\in X^{\mathcal{E}} }^{} \sigma_{\textbf{m}_i}^{\alpha}\right)^{1/\alpha} \right] = \Tr\left[ \left( \textbf{A}_{\alpha}\right)^{1/\alpha} \right]\)
    , where \( \textbf{A}_{\alpha} = \sum_{\textbf{m}_i\in X^{\mathcal{E}} }^{} \sigma_{\textbf{m}_i}^{\alpha}\).
    Since  \(\textbf{A}_{\alpha} \) is Hermitian (sum of Hermitian matrices) it can be diagonalized, thus,
    \begin{equation}\label{eq:traaa} \Tr\left[
    \left(  \textbf{A}_{\alpha}\right)^{1/\alpha}
    \right] = \sum_{i=1}^{2^N} \lambda_i( \textbf{A}_{\alpha}^{1/\alpha})  = \sum_{i=1}^{2^N} [\lambda_i( \textbf{A}_{\alpha})]^{1/\alpha} \leq 2^{N} [\lambda_{\text{max}}( \textbf{A}_{\alpha})]^{1/\alpha}. 
    \end{equation}
    Then, the maximum eigenvalue of \(\mathbf{A}_\alpha\) may be decomposed as \begin{equation}
        \lambda_{\text{max}}( \textbf{A}_{\alpha}) = \lambda_{\text{max}}\left( \sum_{\textbf{m}_i\in X^{\mathcal{E}} }^{} \sigma_{\textbf{m}_i}^{\alpha}\right).
    \end{equation}
    Using Lemma \ref{lemma:weylineq} we have
     \begin{equation}
        \lambda_{\text{max}}( \textbf{A}_{\alpha})\leq \sum_{\textbf{m}_i\in X^{\mathcal{E}} }^{}  \lambda_{\text{max}}(\sigma_{\textbf{m}_i}^{\alpha}) = \sum_{\textbf{m}_i\in X^{\mathcal{E}} }^{}  \left[\lambda_{\text{max}}(\sigma_{\textbf{m}_i})\right]^{\alpha} .
    \end{equation}
    Now, let \(\sigma_{\textbf{m}_{*}}\) be a state whose largest eigenvalue is the maximum over all \(\sigma_{\mathbf{m}_i}\), that is,  \(\lambda_{\text{max}}(\sigma_{\textbf{m}_{*}}) \geq \lambda_{\text{max}}(\sigma_{\textbf{m}_{i}})\) for any other state \(\sigma_{\textbf{m}_{i}}\). As such, 
    \begin{equation}
        \lambda_{\text{max}}( \textbf{A}_{\alpha})\leq 4^{|\mathcal{E}|} [\lambda_{\text{max}}(\sigma_{\textbf{m}_{*}}) ]^{\alpha}.
    \end{equation}
    Considering again Equation~\eqref{eq:traaa}, in turn, means that
    \begin{align}\begin{split}
        \Tr\left[
        \left(  \textbf{A}_{\alpha}\right)^{1/\alpha}
        \right] 
        &\leq 2^{N} [\lambda_{\text{max}}( \textbf{A}_{\alpha})]^{1/\alpha} \\
        &\leq 2^N \left (      4^{|\mathcal{E}|} [\lambda_{\text{max}}(\sigma_{\textbf{m}_{*}}) ]^{\alpha}              \right)^{1/\alpha} \\
        &\leq 2^N      4^{\frac{|\mathcal{E}|}{\alpha}}\lambda_{\text{max}}(\sigma_{\textbf{m}_{*}}). 
    \end{split}\end{align}
    Then, for Equation~\eqref{eq:iaaa} we get
    \begin{align}\begin{split}
        \mathcal{I}_{\alpha} (N) &=  \frac{1}{ |\mathcal{E}| \cdot 2^{N}}\Tr\left[\textbf{A}_{\alpha}^{1/\alpha} \right]\\
        &\leq \frac{1}{ |\mathcal{E}| \cdot 2^{N}}\, 2^N \,  4^{\frac{|\mathcal{E}|}{\alpha}}\,\lambda_{\text{max}}(\sigma_{\textbf{m}_{*}}).
    \end{split}\end{align}
    For \(\alpha \gg |\mathcal{E}|\), we have that \( \mathcal{I}_{\alpha \gg |\mathcal{E}|} (N)  \leq  ({\lambda_{\text{max}}(\sigma_{\textbf{m}_{*}})}/{ |\mathcal{E}| } )4^{\approx 0} \), which with Equation~\eqref{eq:prpiiaaa} yields that
    \begin{equation}
        \Pr^{PI}_{\text{guess}}(x|\mathcal{R}) \leq \frac{\lambda_{\text{max}}(\sigma_{\textbf{m}_{*}})}{|\mathcal{E}|}.
    \end{equation}
    Finally, Lemma~\ref{lemma:lambdasigstar} establishes that
    \(\lambda_{\text{max}}(\sigma_{\textbf{m}_{*}}) \leq N^2/4 + N/4-1/2\), and, by direct substitution, we have
    \begin{equation} 
        \Pr^{PI}_{\text{guess}}(x|\mathcal{R}) \leq \frac{1}{2}+\frac{1}{N}.
    \end{equation}
    Hence, setting \(N = 2^{\sigma}\) makes the \gls{OT} protocol implementation of Figure~\ref{fig:protocol} both \sender{}-secure and \receiver{}-secure, which concludes the proof.
    \end{proof}

\subsubsection{Relinquishing the \(\tau\) Constraint}\label{subsec:noDt}
    The construction from Figure~\ref{fig:protocol} requires that, at one point of the execution, the \sender{} waits for a time interval \(\tau\), such that, given the \gls{NQSM}, the \receiver{} must measure the qubits before receiving the indices \(k,\ell\).
    This constraint might be questioned, as it introduces a substantial delay in the system, specially comparing with the generating and transmission of the required messages (qubits and indices) that can be as fast as the speed of light.
    If the trade-off between the waited time \(\tau\) and the time required to generate and send qubits favors the latter, then this waiting can be removed without affecting the unconditional security, but relaxing the \gls{NQSM} to the \gls{BQSM} instead.
    Indeed, by considering that the \gls{BQSM} forces a limitation on the amount of qubits stored (maximum size of the memory), estimated given some specific limitation of the technology, \(\tau\) can be set to zero.
    Still, a malicious \receiver{} would not be able to cheat and recover more than one of the inputs of the \sender{}, even by measuring its stored system after receiving the indices \(k,\ell\) from the \sender{}.
    
    Note that setting \(\tau=0\) means that the indices (\(k,\ell\)) are sent immediately after the register \(\mathsf{R}\).
    This effectively merges the two messages into an arbitrarily small time period, approaching what could be considered a one-shot protocol.
    However, we  still consider this a two-message procedure, as the messages cannot happen simultaneously (i.e., cannot be permuted), and are inherently sequential with a fixed order (first qubits, then indices), as in the phases of Definition~\ref{def:bqsm}.

    \begin{theorem}\label{thm:security2}
        The protocol from Figure~\ref{fig:protocol} implements a 1-out-of-2 \acrlong{OT} protocol secure against computationally unbounded adversaries (unconditional security parameterized by \(\sigma\)) in the \acrlong{BQSM} with time bound \(t=\tau=0\) and memory bound \(M\).
    \end{theorem}
    \begin{proof}
        As in Theorem~\ref{thm:security1}, start by perceiving that the protocol from Figure~\ref{fig:protocol} implements a 1-out-of-2 \gls{OT}.
        Then, note that the \receiver{}-security (against a malicious \sender{}) does not rely on the \gls{BQSM}, and so this does not alter this part of the security proof.
        Thus, all that requires proving is the \sender{}-security of the protocol, i.e., against a malicious \receiver{}.
    
        In this modified setting, besides the general measurements described in the proof of Theorem~\ref{thm:security1}, there is an added possibility that the \receiver{} performs joint measurements on the system, by storing some of its qubits until after knowing \(k,\ell\).
        From the \gls{BQSM} (Definition~\ref{def:bqsm}), let \(M\) be a parameter representing the maximum size of the memory of a party in the transient phase \(\mathcal{M}_{\tau,M}\).
        Note that, from the Shannon's source coding theorem~\cite{Shannon48}, no unitary can be applied that compresses the \(N\) transmitted qubits into a smaller number, since these are independent and uniformly random prepared states.
        Then, let \(Z\) be the event of sampling \(M\) indices from \(\{1,\dots,N\}\) without replacement (the qubits stored in memory by the \receiver{}), for a security parameter \(\sigma\), set \(N>M\) such that
        \begin{equation}
            \Pr[k,\ell \in Z] = 2\,\frac{M}{N}\,\frac{M-1}{N-1} < 2^{-\sigma}.
        \end{equation}
        Therefore, as long as the phase \(\mathcal{M}_{\tau,M}\) of the \gls{BQSM} happens to the memory of the \receiver{} between receiving register \(\mathsf{R}\) and the indices \(k,\ell\), the receiver can only get one of the inputs of the \sender{}, up to an exponentially low probability \(2^{-\sigma}\), for a large enough \(N\), assuring the security of the \gls{OT}.
    \end{proof}

\section{One-Time Memory and One-Time Programs}\label{sec:CompProtocol}
In this section, the two-message unconditionally-secure \gls{OT} protocol from Section~\ref{sec:ITProtocol} is expanded upon to make it a one-shot \gls{OT}, achieving a \gls{OTM}.
This is accomplished by relaxing the security of the \gls{OT} protocol to rely on computational assumptions (namely, \glspl{TLP} built from \glspl{OWF} and \glspl{SF}), thus enforcing restrictions on the computing capabilities of adversarial parties, and by still working in the \gls{NQSM}.
Still, even though a computational assumption is introduced, the protocol is everlasting secure, as the non-chosen message cannot be retrieved after the execution of the protocol.
Then, by using the compiler from~\cite{C:GKR08}, we achieve \glspl{OTP} from this \gls{OTM}, under the assumptions of existence of a \gls{OWF} and a \gls{SF}, and the \gls{NQSM}.

We start by introducing the concept of \gls{TLP}, a cryptographic primitive whose security relies on computational hardness assumptions (Section~\ref{sec:tlp}).
Then, we leverage this primitive together with the previous construction of Section~\ref{subsec:NIOT} to achieve the desired one-shot \gls{OT} protocol, i.e., the \gls{OTM} (Section~\ref{sec:1sot}).
Lastly, as a corollary, we use the compiler from~\cite{C:GKR08} to get \glspl{OTP} from the \gls{OTM} construction (Section~\ref{subsec:otp}).

\subsection{Time-Lock Puzzles}\label{sec:tlp}
A \gls{TLP}~\cite{RSW96} is a non-interactive cryptographic primitive that allows for a party to send a hidden message, such that this message can only be read after some time has elapsed.
It is required that a puzzle can be efficiently generated, i.e., the time to generate the puzzle must be much less than the time to solve it; and that the secret can only be read after some pre-defined time, even for parallel algorithms.
Definitions~\ref{def:puzzle}~and~\ref{def:tlp} formally state this idea.
The minimal assumptions required to realize a \gls{TLP} have been studied in~\cite{IC:JMRR21}.

\glspl{TLP} have a wide variety of applications, but in this work they will be integrated in the \gls{NQSM} to introduce a delay in the protocol, such that the quantum memory of a party will decohere before it is able to access the information hidden by the \gls{TLP}.

\begin{definition}[Puzzle~\cite{ITCS:BGJ+27}]\label{def:puzzle}
    Let \(\lambda\in\mathbb{N}\) be the security parameter. A \textit{puzzle} is a pair of algorithms \((\mathsf{Puzzle.Gen}\), \(\mathsf{Puzzle.Sol})\) with
    \begin{itemize}
        \item \(Z \leftarrow \mathsf{Puzzle.Gen}(\tau, s)\) takes as input a time parameter \(\tau\) and a solution \(s \in \{0, 1\}^\lambda\), and outputs a puzzle \(Z\). \(\mathsf{Puzzle.Gen}(\tau, s)\) takes \(\operatorname{poly}(\log \tau, \lambda)\) time.
        \item \(s \leftarrow \mathsf{Puzzle.Sol}(Z)\) takes as input a puzzle \(Z\) and outputs a solution \(s\). \(\mathsf{Puzzle.Sol}(Z)\) takes \(\tau\cdot poly(\lambda)\) time.
    \end{itemize}
    Then, for all \(\lambda\), time parameter \(\tau\), solution \(s \in \{0, 1\}^\lambda\), and puzzle \(Z\) in the support of \(\mathsf{Puzzle.Gen}(\tau, s)\), \(\mathsf{Puzzle.Sol}(Z)\) outputs \(s\).
\end{definition}

\begin{definition}[Time-Lock Puzzle~\cite{ITCS:BGJ+27}]\label{def:tlp}
    A puzzle \((\mathsf{Puzzle.Gen}, \mathsf{Puzzle.Sol})\) is a \textit{time-lock puzzle} with gap \(\varepsilon < 1\) if there exists a polynomial \(\underline{\tau}(\cdot)\), such that for every polynomial \(\tau(\cdot) \geq \underline{\tau}(\cdot)\) and adversary \(\mathcal{A} = \left\{\mathcal{A}_\lambda \right\}_{\lambda \in \mathbb{N}}\) of depth smaller than \(\tau^\varepsilon(\lambda)\), there exists a negligible function \(\mu\), such that for all \(\lambda \in \mathbb{N}\) and \(s_0, s_1 \in \{0, 1\}^\lambda\):
    \[\Pr\left[b\leftarrow \mathcal{A}_\lambda(Z):
    \begin{array}{l}
    b\leftarrow \{0,1\}\\
    Z \leftarrow \mathsf{Puzzle.Gen}(\tau(\lambda), s_b)
    \end{array}
    \right] \leq \frac{1}{2} + \mu(\lambda). \]
\end{definition}

In this work, minimal requirements for the \glspl{TLP} are needed.
In particular, it is enough to consider \textit{weak \acrlongpl{TLP}}~\cite{ITCS:BGJ+27} that can be build directly from \glspl{OWF} (assuming the existence of a non-parallelizing language\footnote{A non-parallelizing language is equivalent to a sequential function~\cite{IC:JMRR21}.}).
This relaxed formulation of \glspl{TLP} only requires that the puzzle can be generated in fast parallel time (circuit computing \textsf{Puzzle.Gen} of size \(\operatorname{poly}(\tau, \lambda)\) has depth \(\operatorname{poly}(\log \tau,\lambda)\)), while it still takes time \(\tau\) to solve (\textsf{Puzzle.Sol} takes time \(\tau\cdot \operatorname{poly}(\lambda)\)).
\begin{lemma}[\cite{ITCS:BGJ+27,IC:JMRR21}]\label{lemma:wtlp}
    There exists a weak \acrlong{TLP}, assuming the existence of a \acrlong{OWF} and a \acrlong{SF}, which fulfills the security definition of Definition~\ref{def:tlp}. 
\end{lemma}

In addition, for our purpose, since the time intervals that are considered in the \gls{NQSM} are often short enough (e.g., \(0.25\si{ms}\)~\cite{NPJ:VAV+22}), the requirements on the puzzle generation can even be further relaxed, such that the time to generate the puzzle may be the same as the time to solve it.
This enables very simple and diverse constructions, such as repeated hashing of a shared seed.
Nevertheless, to be as general as possible and limit the setup assumptions to the \gls{NQSM}, without imposing conditions on its parameters (time to quantum decoherence), weak \glspl{TLP} are considered from here onwards.

\subsection{One-Time Memory (or, one-shot Oblivious Transfer)}\label{sec:1sot}
Here, a construction for a one-shot chosen-bit 1-out-of-2 \gls{OT} is given, yielding the desired \gls{OTM}.
First, in Section~\ref{sec:ITProtocol}, a two-message non-interactive unconditionally-secure 1-out-of-2 \gls{OT} in the \gls{NQSM} was described.
Now, by using a \gls{OWF} and a \gls{SF} via a \gls{TLP} in the protocol, relaxing the security requirements to hold on computationally-hard problems, a one-shot 1-out-of-2 \gls{OT}, i.e., a \gls{OTM}, is constructed.

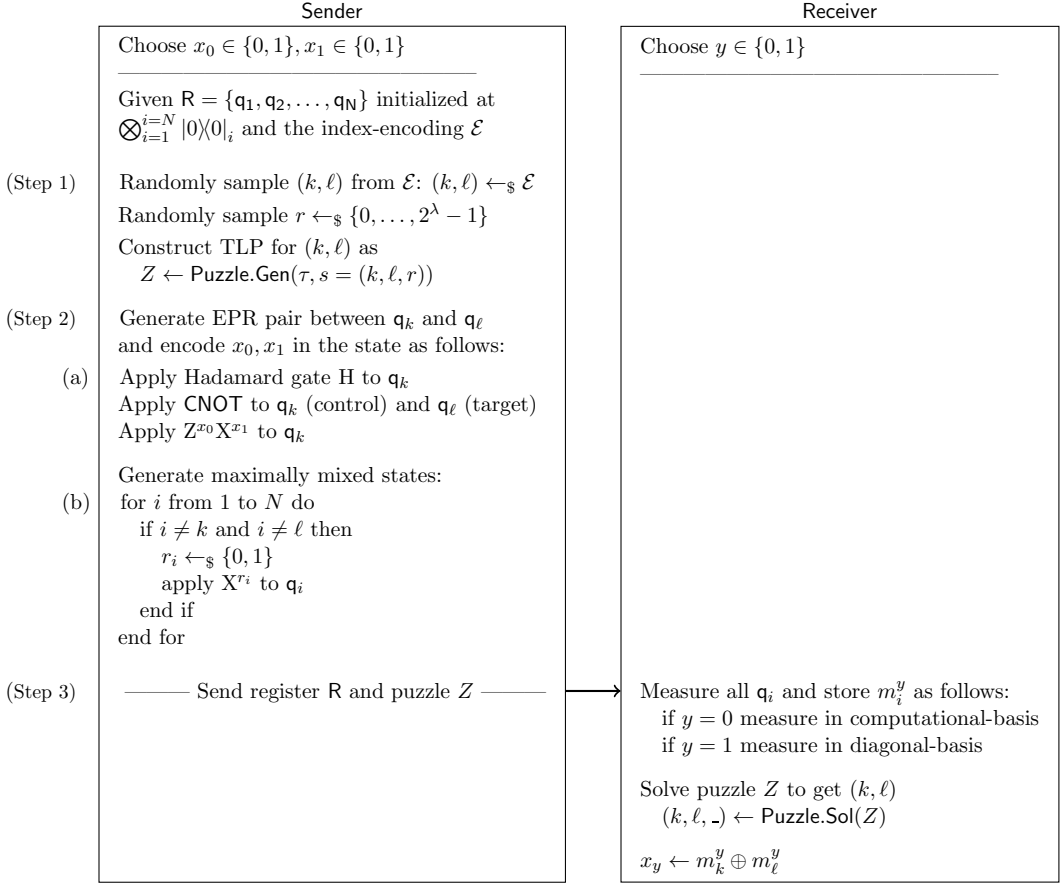
\begin{figure}[htb]
    \centering
    \input{protocolnew1rX} 
    \caption{Schematic representation of the proposed one-shot \gls{OT} protocol parameterized by \(N(\sigma), \lambda, \tau\).}
    \label{fig:protocol1r}
\end{figure}

The operational description of the \gls{OTM} is given in Figure~\ref{fig:protocol1r}, and executes analogously to the protocol from Section~\ref{sec:ITProtocol}.
Below, we detail the differences in the various steps when compared to the previous one. Step 0 and Step 2, which are not explicitly mentioned, are identical to Figure \ref{fig:protocol}.

    \begin{itemize}
        \item \textbf{Step 1}: 
        The \sender{} uniformly samples indices \(k,\ell\) from the index encoding set \(\mathcal{E}\) (with \(k<\ell\), without loss of generality), selecting qubits \(\{\mathsf{q}_k,\mathsf{q}_\ell\} \subset \mathsf{R}\). The \sender{} hides the  \(k,\ell\), as the solution of the \gls{TLP} (\(Z\)), parametrized by \(\tau\) whose lower bound is established by the \gls{NQSM}.
        \item \textbf{Step 3}: The \sender{} sends the entire register \(\mathsf{R}\) and the \gls{TLP} (\(Z\)) to the \receiver{}. The \receiver{} measures each individual qubit \(\mathsf{q}_i\), either in  computational basis if \(y=0\) or in the diagonal basis \(y=1\),  and stores all classical measurement results \(m^y_i\in \{0,1\}\). Concurrently, the \receiver{} solves the \gls{TLP} (\(Z\)), which will reveal the indices \(k,\ell\) as the solution. Finally, once the puzzle is solved, the \receiver{} computes the parity of the stored measurement outputs for \(\{\mathsf{q}_k, \mathsf{q}_\ell\}\), that is, \( m^y_k\oplus m^y_\ell = x_y\).
    \end{itemize}

The protocol still works in the \gls{NQSM}, but instead of relying on an explicit time-delay introduced by the \sender{} in the execution of the protocol, it relies on a \gls{TLP} to enforce it.
This has several advantages (besides proving that such a construction is possible), as it delegates the responsibility of time-keeping from the sender to a cryptographic primitive.
But, perhaps as important, it allows for a single \gls{TLP} to hide the secret information of many \glspl{OT}/\glspl{OTM}, effectively amortizing the time lag and computation required to perform many executions that are performed in parallel, greatly boosting performance.

For the protocol to be secure, the \gls{TLP} is designed such that it explores the quantum decoherence of imperfect quantum memories, here embodied by the \gls{NQSM}.
Setting the time it takes to solve the \gls{TLP} (\(\tau\)) such that it  is larger than the decoherence time modeled by the \gls{NQSM}, again, enforces the \receiver{} to measures the two entangled qubits without knowing the encoding, as required to achieve security.

\subsubsection{Security}
Again, to guarantee security, it must be proved that the \sender{} cannot obtain any information regarding the input of the \receiver{}; and, that the \receiver{} can recover at most one of the inputs of the \sender{}.
Since this is a one-shot protocol, security requires that: the \sender{} cannot construct a message (e.g., by keeping correlated ancillas) that allows it to extract any information on the input of the \receiver{}; and that a (single) honestly-crafted message does not reveal more than one of the inputs of the \sender{} regardless of any \gls{POVM} on the overall register that the \receiver{} can perform, and assuming the security of the underlying assumptions of the \gls{TLP}.

The protocol is parameterized by the statistical security parameter \(\sigma\), computational security parameter \(\lambda\), and the time \(\tau\) to quantum decoherence of memories established from the \gls{NQSM}.

\begin{theorem}\label{thm:security1r}
    The protocol from Figure~\ref{fig:protocol1r} implements a computationally-secure 1-out-of-2 \acrlong{OT} protocol, assuming the existence of a \acrlong{OWF} and a \acrlong{SF}, in the \acrlong{NQSM} (parameterized by \(\sigma, \lambda, \tau\)). It is a \acrlong{OTM}, since it is also one-shot.
\end{theorem}
\begin{proof}
Assuredly, the protocol of Figure~\ref{fig:protocol1r} implements a 1-out-of-2 \gls{OT} functionality.
It is also one-shot, thus if it implements a secure (chosen-input) 1-out-of-2 \gls{OT}, then it is a \gls{OTM} (Definition~\ref{def:otm}).
So, it remains to prove that it fulfills the security requirements of \gls{OT} in the \gls{NQSM}.

From Lemma~\ref{lemma:wtlp}, there exists a secure weak \gls{TLP} assuming the existence of a \gls{OWF} and a \gls{SF}, which can be generated in parallel in time \(\log \tau\) and that takes time \(\tau\) to solve.
Then, from the \gls{NQSM}, let \(\tau\) be the time that a quantum memory takes to completely decohere, up to probability \(2^{-\sigma}\).
Again, the \gls{NQSM} can be applied to the setting of this protocol as the memory of a malicious \receiver{} must linearly increase with \(N\), the number of sent qubits by the \sender{}, which exponentially decreases its memory storage capabilities, as in Equation~\eqref{eq:nqsm}.

\paragraph{\receiver{}-security:} 
Same as in Section~\ref{subsec:security}.
All the sender does is send the same (\(N\)) qubits as before, and instead of sending the indices \(k,\ell\) after, it sends a \gls{TLP} hiding \(k,\ell\) together with the qubits.
Clearly, from the security of the weak \gls{TLP}, there is nothing the \sender{} can do that allow it to gain any information on the input of the \receiver{}.

\paragraph{\sender{}-security:}
From the security of the \gls{TLP}, the puzzle does not reveal any information about the indices \(k,\ell\) before time \(\tau\), up to negligible probability in \({\lambda}\).
Assuming the \gls{NQSM}, this means that a malicious \receiver{} cannot store the \(N\) qubits more time than the one it takes to solve the puzzle, as they would completely decohere.
Then, before time \(\tau\), the view of the \receiver{} is indistinguishable (it is the same) of its view in the previous setting of Section~\ref{subsec:security} (where all the \receiver{} sees is the \(N\) qubits, before receiving \(k,\ell\)), up to a negligible probability in \(\lambda\), assuming the hardness of the weak \gls{TLP}.
And, after time \(\tau\), the view is also indistinguishable, as in both cases the \receiver{} gets total information on the indices \(k,\ell\).
Thus, all a malicious \receiver{} can do in this setting, it could also do in the secure setting of Section~\ref{subsec:NIOT}, which is proved to be secure.
Moreover, since after the execution of the protocol (after time \(\tau\)) all information in the system is lost, there is no possibility of the malicious \receiver{} recovering the other message in the future, making the \gls{OT} protocol everlasting secure.

\medskip

Therefore, by reduction, assuming the existence of a \gls{OWF} and a \gls{SF}, and working in the \gls{NQSM}, no malicious \sender{} or malicious \receiver{} can do anything more when engaging in the protocol of Figure~\ref{fig:protocol1r} than they could have done in the secure protocol established in Section~\ref{subsec:NIOT}.
The proof that the protocol of Figure~\ref{fig:protocol1r} implements a secure chosen-input 1-out-of-2 \gls{OT}, together with the property of being one-shot (consists of a single message transmitted from the \sender{} to the \receiver{}) means that this protocol is a \gls{OTM}.
\end{proof}

\subsection{One-Time Program}\label{subsec:otp}
We construct \glspl{OTP} in the \gls{NQSM}, by using the designed \gls{OTM} of Section~\ref{sec:1sot}, and thus assuming the existence of a \gls{OWF} and a \gls{SF}. 
This result comes as a direct corollary of the one-time compiler from~\cite{C:GKR08,TCC:GIS+10}, where from parallel-\gls{OTM} and \glspl{OWF}, it is possible to construct \glspl{OTP} from any standard program (Theorem~\ref{thm:otp-compiler}).
This compiler is directly applicable to the polynomial parallel executions of the \gls{OTM} achieved in Section~\ref{sec:1sot}.
As such, its size is exponential (but constant) in the statistical security parameter, but still only grows polynomially with the power of the adversary (i.e., with the computational security parameter).

Indeed, our \gls{OTM} construction, given the essence of the \gls{NQSM}, has the property that it must be evaluated right away.
Unfortunately this is unavoidable in our setting of the \gls{NQSM}, but still allows for \glspl{OTP} that are evaluated as they are received.
Critically, our construction of \gls{OTM} implies that, once received, all its qubits must be measured before the \gls{TLP} may be solved (see Theorem~\ref{thm:security1r}).
From this fact, it is straightforward to notice that our construction of \gls{OTM} immediately yields parallel-\gls{OTM}:
Suppose an adversary receives two \glspl{OTM}, implemented as in Section~\ref{sec:1sot}.
If it waits to finish solving the \gls{TLP} corresponding to the first before measuring the states corresponding to the second, these would have completely decohered by then (up to exponentially low probability).
Moreover, the \gls{OTM} protocol of Figure~\ref{fig:protocol1r} may also be trivially extended to encompass the hiding of all encoding information of the multiple \glspl{OTM} in a single \gls{TLP}, making the proof of parallel-\gls{OTM} follow by direct application of the proof of Theorem~\ref{thm:security1r}.
Note that, even though we propose \glspl{OTM} with bit outputs, these may be compiled to \gls{OTM} with string outputs, and then directly applied to obtain \gls{OTP}~\cite{TCC:GIS+10}.

We remark that the definitions introduced in~\cite{C:GKR08,TCC:GIS+10} are simulation-based definitions, while in our work we prove security by demonstrating the properties of \sender{} and \receiver{} security of \gls{OT}, together with the property of being one-shot, according to Definitions~\ref{def:ot} and~\ref{def:otm}.
Nevertheless, assuming our proposal described in Figure~\ref{fig:protocol1r} implements an \gls{OTM}, the simulation based security of the compiled \gls{OTP} is independent of the underlying \gls{OTM} implementation.

\section*{Acknowledgements}
The authors thank David Elkouss for insightful discussions.

\medskip

\noindent
RF acknowledges the support of the QuantaGenomics project funded within the QuantERA II Programme that has received funding from the European Union's Horizon 2020 research and innovation programme under Grant Agreement No 101017733, and with funding organisations, The Foundation for Science and Technology – FCT (QuantERA/0001/2021), Agence Nationale de la Recherche - ANR, and State Research Agency – AEI. This work was supported in part by the European Union under the programs Horizon Europe R\&I, through the project QSNP (GA 101114043). 
MG acknowledges FCT - Fundação para a Ciência e a Tecnologia (Portugal) financing refs. UIDB/50021/2020 and UIDP/50021/2020 (resp. DOI 10.54499/UIDB/50021/2020 and 10.54499/UIDP/50021/2020).
LN acknowledges support from FCT - Fundação para a Ciência e a Tecnologia (Portugal) via the Project No. CEECINST/00062/2018.
EZC acknowledges funding by FCT/MCTES - Fundação para a Ciência e a Tecnologia (Portugal) - through national funds and when applicable co-funding by EU funds under the project UIDB/50008/2020, and funding by FCT through project 2021.03707.CEECIND/CP1653/CT0002.

\bibliographystyle{alpha}
\bibliography{bib}

\end{document}

%% file: acronyms.tex
\newacronym{2PC}{2PC}{Two-Party Computation}
\newacronym{BC}{BC}{Bit Commitment}
\newacronym{BQSM}{BQSM}{Bounded-Quantum-Storage Model}
\newacronym{CPTP}{CPTP}{Completely Positive Trace Preserving}
\newacronym{DI}{DI}{Device-Independent}
\newacronym{EA-PM}{EA-PM}{Entanglement-Assisted Prepare-and-Measure}
\newacronym{EPR}{EPR}{Einstein-Podolsky-Rosen}
\newacronym{MPC}{MPC}{Multi-Party Computation}
\newacronym{NISQ}{NISQ}{Noisy Intermediate Scale Quantum}
\newacronym{NPA}{NPA}{Navascués-Pironio-Acín}
\newacronym{NQSM}{NQSM}{Noisy-Quantum-Storage Model}
\newacronym{OT}{OT}{Oblivious Transfer}
\newacronym{OTP}{OTP}{One-Time Program}
\newacronym{OTM}{OTM}{One-Time Memory}
\newacronym{OWF}{OWF}{One-Way Function}
\newacronym{PKC}{PKC}{Public-Key Cryptography}
\newacronym{PM}{PM}{Prepare-and-Measure}
\newacronym{POVM}{POVM}{Positive Operator-Valued Measure}
\newacronym{QROM}{QROM}{Quantum Random Oracle Model}
\newacronym{RAC}{RAC}{Random Access Code}
\newacronym{SDC}{SDC}{Superdense Coding}
\newacronym{SF}{SF}{Sequential Function}
\newacronym{TLP}{TLP}{Time-Lock Puzzle}

%% file: protocolnewX.tex
\centering
\begin{tikzpicture}[scale=0.75, transform shape, node distance=9cm, auto]
    \node[draw, rectangle, minimum width=5cm, minimum height=13cm, label=above:\textsf{Sender}] (sender) {
        \begin{tabular}{l}
            Choose $x_0 \in \{0,1\}, x_1 \in \{0,1\}$ \\ 
           --------------------------------------------------\\

            Given $\mathsf{R = \{q_1, q_2, \ldots, q_N\}}$ initialized at \\ \(\bigotimes_{i=1}^{i=N}\ketbra{0}{0}_i\) and the index-encoding \(\mathcal{E}\) \\\\

           {\hspace{-2cm}(\small Step 1)} \hspace{.5cm} Randomly sample \((k,\ell)\) from \(\mathcal{E}\): \((k,\ell)\leftarrow_{\$} \mathcal{E}\)\\ \\
         
           {\hspace{-2cm}(\small Step 2)} \hspace{.5cm} Generate EPR pair between \(\mathsf{q}_k\) and \(\mathsf{q}_\ell\)  \\and encode \(x_0,x_1\) in the state as follows: \vspace{0.1cm} \\ 
           
           {\hspace{-1cm}(a)}  \hspace{.4cm}Apply Hadamard gate \(\mathrm{H}\) to \(\mathsf{q}_k\) \\ 
           Apply \(\mathsf{CNOT}\) to \(\mathsf{q}_k\) (control) and \(\mathsf{q}_\ell\) (target)\\  
           Apply \(\mathrm{Z}^{x_0}\mathrm{X}^{x_1} \) to \(\mathsf{q}_k\) \vspace{0.3cm} \\  Generate maximally mixed states:\\
           {\hspace{-1cm}(b)}  \hspace{.4cm}for $i$ from $1$ to $N$ do  \\ \quad if $i \neq k$ and $i \neq \ell$ then \\
           \quad\quad \(r_i\leftarrow_{\$} \{0,1\}\) 
           \\\quad\quad apply $\mathrm{X}^{r_i}  $ to \(\mathsf{q}_i\) \\ \quad end if \\  end for\\ \\
         {\hspace{-2cm}(\small Step 3)} \hspace{.6cm}   ---------------Send register \(\mathsf{R}\)  -----------------\\ \\
          {\hspace{-2cm}(\small Step 4)} \hspace{.5cm}   \hspace{2.5cm}  Wait $\tau$ \\   \\ 
           {\hspace{-2cm}(\small Step 5)} \hspace{.6cm}  ---------------Send indices \(k,\ell\)  -----------------\\ \\
        \end{tabular}
    };
    
    \node[draw, rectangle, minimum width=5cm, minimum height=13cm, label=above:\textsf{Receiver}, right of=sender] (receiver) {
        \begin{tabular}{l}
           Choose $y\in \{0,1\}$ \\ \vspace{9cm}  
           -------------------------------------------------- \\ 

        Measure all $\mathsf{q}_i$ and store  \(m^{y}_i\) as follows:\\
        \quad  if $y = 0$ measure in computational-basis \\
        \quad if $y = 1$ measure in diagonal-basis\\[1em]
        $x_y\leftarrow m^{y}_k \oplus m^{y}_\ell$
        \end{tabular}
    };

    \draw[thick,->] ([xshift=0cm, yshift=-3.8cm]sender.east) -- ([yshift=-3.8cm]receiver.west);
    \draw[thick,->] ([xshift=0cm,yshift=-5.7cm]sender.east) -- ([yshift=-5.7cm]receiver.west);

\end{tikzpicture}

%% file: protocolnew1rX.tex
\centering
\begin{tikzpicture}[scale=0.75, transform shape, node distance=9cm, auto]
    \node[draw, rectangle, minimum width=5cm, minimum height=15cm, label=above:\textsf{Sender}] (sender) {
        \begin{tabular}{l}
            Choose $x_0 \in \{0,1\}, x_1 \in \{0,1\}$ \\ 
           --------------------------------------------------\\

            Given $\mathsf{R = \{q_1, q_2, \ldots, q_N\}}$ initialized at \\ \(\bigotimes_{i=1}^{i=N}\ketbra{0}{0}_i\) and the index-encoding \(\mathcal{E}\) \\\\

           {\hspace{-2cm}(\small Step 1)} \hspace{.5cm} Randomly sample \((k,\ell)\) from \(\mathcal{E}\): \((k,\ell)\leftarrow_{\$} \mathcal{E}\)\\[.1cm]
           Randomly sample \(r\leftarrow_\$ \{0,\dots,2^\lambda -1\}\)\\[.1cm]
           Construct TLP for \((k,\ell)\) as \\
           \quad\(Z\leftarrow \mathsf{Puzzle.Gen}(\tau, s=(k,\ell,r))\)\\[.3cm]
         
           {\hspace{-2cm}(\small Step 2)} \hspace{.5cm} Generate EPR pair between \(\mathsf{q}_k\) and \(\mathsf{q}_\ell\)  \\and encode \(x_0,x_1\) in the state as follows: \vspace{0.1cm} \\ 
           
           {\hspace{-1cm}(a)}  \hspace{.4cm}Apply Hadamard gate \(\mathrm{H}\) to \(\mathsf{q}_k\) \\ 
           Apply \(\mathsf{CNOT}\) to \(\mathsf{q}_k\) (control) and \(\mathsf{q}_\ell\) (target)\\  
           Apply \(\mathrm{Z}^{x_0}\mathrm{X}^{x_1} \) to \(\mathsf{q}_k\) \vspace{0.3cm} \\  Generate maximally mixed states:\\
           {\hspace{-1cm}(b)}  \hspace{.4cm}for $i$ from $1$ to $N$ do  \\ \quad if $i \neq k$ and $i \neq \ell$ then \\
           \quad\quad \(r_i\leftarrow_{\$} \{0,1\}\) 
           \\\quad\quad apply $\mathrm{X}^{r_i}  $ to \(\mathsf{q}_i\) \\ \quad end if \\  end for\\ \\
         {\hspace{-2cm}(\small Step 3)} \hspace{.6cm}   --------- Send register \(\mathsf{R}\) and puzzle \(Z\)  --------- \\[3cm]
        \end{tabular}
    };
    
    \node[draw, rectangle, minimum width=5cm, minimum height=15cm, label=above:\textsf{Receiver}, right of=sender] (receiver) {
        \begin{tabular}{l}
           Choose $y\in \{0,1\}$ \\ \vspace{10.5cm} 
           -------------------------------------------------- \\

        Measure all $\mathsf{q}_i$ and store  \(m^{y}_i\) as follows:\\
          \quad  if $y = 0$ measure in computational-basis \\
           \quad if $y = 1$ measure in diagonal-basis  \\[.3cm]
        Solve puzzle \(Z\) to get \((k,\ell)\)\\
        \quad \((k,\ell, \_) \leftarrow \mathsf{Puzzle.Sol}(Z)\)\\[.3cm]
        $x_y\leftarrow m^{y}_k \oplus m^{y}_\ell$
        \end{tabular}
    };

    \draw[thick,->] ([xshift=0cm,yshift=-4.2cm]sender.east) -- ([yshift=-4.2cm]receiver.west);
\end{tikzpicture}